%% file: SISO-mehran_V2.tex
\documentclass[twocolumn, journal]{IEEEtran}
\input{pkgs.tex}

\input{defs.tex}
\usepackage{cite}
\usepackage{balance}

\ifCLASSINFOpdf
\else
\fi
\begin{document}
\title{On Topological Properties of the Set of\\Stabilizing Feedback Gains}

\author{Jingjing Bu,  \and Afshin Mesbahi, \and Mehran Mesbahi\thanks{The authors are with the University of Washington; Emails: {\em bu+amesbahi+mesbahi@uw.edu}}}

\maketitle

\begin{abstract}
This work presents a fairly complete account on various topological and metrical aspects of feedback stabilization for single-input-single-output (SISO) continuous and discrete time linear-time-invariant (LTI) systems. In particular, we prove that the set of stabilizing output feedback gains for a SISO system with $n$ states has at most $\ceil{\frac{n}{2}}$ connected components. Furthermore, our analysis yields an algorithm for determining intervals of stabilizing gains for general continuous and discrete LIT systems; the proposed algorithm also computes the number of unstable roots in each unstable interval. Along the way, we also make a number of observations on the set of stabilizing state feedback gains for MIMO systems.
\end{abstract}


% \fi
\IEEEpeerreviewmaketitle

\section{Introduction} 
In classical control, topological and metrical properties of stabilizing feedback gains are of paramount importance for the stability analysis and stabilization of LIT systems \cite{Vidyasagar_TAC_1982,Byrnes_book_1980, ohara1992differential,Hermann_TAC_1977no}. % Traditionally, feedback stabilization is parameterized by transfer functions \cite{Vidyasagar_TAC_1982, Byrnes_book_1980, Hermann_TAC_1977no}.
Recently, such 
%a great deal of effort has been made for
properties 
have received renewed interest in system literature as they 
have direct implications for adopting learning algorithms for control design.
This is particularly the case in the so-called direct policy algorithms, where it is of interest to directly adjust the control gain--without an explicit system identification step--say, using a gradient step.
The design objectives in these scenarios are typically functions over direct policies--often desired to be stabilizing feedback gains.
%in the matrix space.
% So, the properties are crucial in understanding cost functions and designing algorithms.
For example, in reinforcement learning, the policy gradient updates the feedback iteratively to get desired optimal controller. In this case, the cost functions are defined on the set of stabilizing controllers (assuming $+\infty$ elsewhere).\footnote{In this paper, we will use ``feedback controllers'' interchangeably with static feedback gains as ``dynamic'' controllers are not considered.} As such, understanding the topological and metrical properties of this set provides valuable insights in designing learning algorithms for dynamic systems. 
In the meantime, such insights can also reveal fundamental  shortcomings in certain optimization algorithms. For example, if the set of stabilizing feedback gains has several path-connected components, the solutions of gradient-type learning algorithms will be highly dependent on the initialization process.

It is thus surprising that despite the long historical interest in characterizing the set of stabilizing feedback gains, research works on its set-theoretic and topological properties are rather limited. This is potentially due to significantly more interest in characterizing the set of 	``certificates" for stabilizing controllers, e.g., in terms of linear matrix inequalities.

Of particular relevance to our work in directly characterizing
stabilizing feedback gains is that of
Ohara and colleagues~\cite{ohara1992differential}, 
who examined its differential geometric structure for multiple-input-multiple-output (MIMO) systems. 
In~\cite{Prakash_IJC_1983} and~\cite{Fam_TAC_1978},
an elegant geometric approach has been adopted to parametrize the set of stabilizing feedback gains; in particular,
it has been shown that for continuous and discrete single-input-single-output (SISO) and dyadic systems the corresponding sets can be bounded via two and three hyperplanes.
%, respectively~\cite{Prakash_IJC_1983} and~\cite{Fam_TAC_1978}.
%
Furthermore, the work of Ober~\cite{ober1987topology} has
shown that the set of stable SISO systems of order $n$ have $n+1$ connected components in the Euclidean topology 
while the set of stable MIMO systems is connected.
The work reported in~\cite{feng2018on} focuses on the connectedness of this set for both SISO and MIMO systems. We note that both works~\cite{ohara1992differential}~\cite{feng2018on} examine continuous-time systems.
%

%In this manuscript, we present topological and metrical properties of the set of stabilizing feedback gains in both continuous and discrete linear time-invariant (LTI) SISO systems. We discuss the openness, boundedness, convexity and connectedness in both continuous- and discrete-time systems. These properties are relevant in understanding cost functions defined on the set and algorithm design. Among these results, we prove that the set $\ca H$ and $\ca S$ both can have at most $\ceil{\frac{n}{2}}$ connected components where $n$ is the number of states. This bound is tight by our simulation results. Our analysis for the bounds reveals an algorithm to identify these path-connected components. If a local search algorithm is sophisticated enough to guarantee optimality on each path-connected component, combining with our initialization scheme, a global optimality will be achieved.
%
This paper discusses the topological, metrical, and geometric properties of the set of stabilizing controllers for both continuous and discrete-time LTI systems.
%\begin{itemize}
 % \item 
We show that the set of stabilizing state-feedback gains for a continuous SISO system is regular open, unbounded, in general nonconvex, and path-connected in the Euclidean topology. In the meantime, the set of stabilizing output-feedback controllers is shown to be open but not connected in general, and can be bounded or unbounded. 
It recent works, based on the implicit assumption that stable and unstable intervals
of the feedback gain interlace, it has been stated that 
the set of stabilizing output feedback controllers for SISO systems can have at most $n$ (in~\cite{feng2018on})
and $\ceil{\frac{n}{2}}$ (in \cite{feng2019}) connected components.
%
%The proofs reported in these works in the meantime, rely on an assumptions that does not hold for general continuous LTI SISO systems. 
If this assumption does not hold, however, the line of reasoning reported in~\cite{feng2018on, feng2019} lead
to the {\em upper bounds} of $2n$ and $n$, respectively.\footnote{In discussions with authors of~\cite{feng2019}, it has been  pointed out that a perturbation type argument can help address this issue; the approach adopted in this work is direct and leads to a constructive algorithm for characterizing the stabilizing and unstabilizing intervals.}
In this work, we prove a tight bound of $\ceil{\frac{n}{2}}$ for continuous as well as discrete time LTI systems; all of our results are constructive (they lead to algorithms for characterizing these sets) and rely on basic topology and analytic theory of polynomials~\cite{marden1966geometry, rahman2002analytic}.\footnote{Thus the emphasis on the SISO case; some of these results have been extended to MIMO case in~\cite{MIMO2019}.}

The separate treatment for continuous and discrete time systems is warranted;
in fact, in contrast to the folklore expectation of unified properties for continuous and discrete time systems, there are counterexamples to show that the analogies between the two are
far from complete~\cite{Jonckheere_TAC_1981, Wimmer_JMSEC_1995}.
The distinct difference between continuous and discrete LTI systems might be due to the fact that the generalized bilinear transform has poles and thus not continuous~\cite{Mehrmann_LAA_1996,Hung_LAA_1998}. Therefore, generalizing the proposed topological properties of the set of stabilizing feedback gains from continuous LTI systems to discrete ones is not straightforward. Nevertheless, in this paper we show that the set of stabilizing state feedback gains for discrete-time LTI SISO systems enjoys some of the topological properties as its continuous counterpart, i.e., open and path-connected in Euclidean topology and nonconvexity. But in contrast to the continuous case, the set of stabilizing state feedback gains is \emph{bounded}. For output feedback SISO systems, the corresponding set of stabilizing gains is open, bounded and in general nonconvex, but is no longer path-connected. Accordingly, we prove that the set can have at most $\ceil{\frac{n}{2}}$ path-connected components, which is a tight bound supported by simulation results.
% Mind that in output-feedback systems, the proofs between continuous- and discrete-time have nominal differences. 
      %\item
       {The present work also proposes an algorithm for determining the intervals of stabilizing feedback gains for general continuous and discrete LIT systems.
    %   , based on the obtained topological and metrical properties. 
       This algorithm also computes the number of unstable roots in each unstable interval.}%We point out algorithms to identify the stabilizing intervals in both continuous- and discrete-time LTI SISO systems.
      %  \item 
%      {\color{red}We demonstrate an application of the algorithms in SISO LQR control.}
%    \end{itemize}
% via \DeclareGraphicsExtensions.
%\caption{Simulation results for the network.}
%\label{fig_sim}
% \end{figure}

      \par
      The paper is organized as follows. In \S\ref{sec:notations}, we introduce the notation and the preliminary background; \S\ref{sec:hurwitz} \S\ref{sec:schur} are devoted to set-theoretic properties of Hurwitz and Schur stabilizing feedback gains, respectively, followed by numerical examples. Our results are intermingled with observations and remarks that further provide insights into some of the geometric and topological intricacies of feedback stabilization.

\section{Notation and Preliminaries}
\label{sec:notations}
% Linear Algebra
We denote by ${\bb M}_n(\bb R)$ the set of $n \times n$ real matrices and ${\bb {GL}}_n(\bb R)$ as its subset of invertible matrices; $\bb R^n$ and $\bb C^n$ denote the $n$-dimensional real and complex Euclidean spaces with $n=1$ identified with real and complex numbers.
 For a vector $v \in \bb R^n$, we use $v_j$ to denote the $j$th entry of $v$, where $v = (v_1, \dots, v_n)^\top$. %\par
% Complex analysis
We denote the open unit disk of $\bb C$ by $\bb D = \{\lambda \in \bb C: |\lambda| < 1\}$ and the left-half plane by $\bb H_{-} = \{\lambda \in \bb C: \text{\bf Re}(\lambda) < 0\}$;\footnote{ $\text{\bf Re}$ and $\text{\bf Im}$ refer to the real and imaginary parts of the complex number; when applied to a set of complex numbers, these operations are naturally extended to each element of that set.} $\bb H^{n}_{-}$ will be the $n$-dimensional version of $\bb H_{-}$. The notation $|\lambda|$ denotes the modulus of the complex number $\lambda \in \bb C$ and $\bar{\lambda}$ denotes its complex conjugate. We use $\bb C[\lambda]$ and $\bb R[\lambda]$ to denote polynomials with complex and real coefficients, respectively, where $\lambda$ is the corresponding indeterminant of the polynomial. For a polynomial $p$ over $\bb C$ or $\bb R$ 
%(\lambda) \in \bb C[\lambda]$ or $p(\lambda) \in \bb R[\lambda]$, 
we use $p'$ to denote its derivative with respect to the indeterminant, unless noted otherwise.
%\par
% Polynomials
By Fundamental Theorem of Algebra, a monic polynomial $p(\lambda) = \lambda^n + \alpha_{n-1} \lambda^{n-1} + \dots + \alpha_0 \in {\bb C}[\lambda]$ has $n$ roots (or zeros) counting multiplicities; we let ${\cal Z}_p$ to denote this set  (each zero is repeated according to its multiplicity). Mind that ${\cal Z}_p$ is not a well-defined object in $\bb C^n$ as we do not impose a natural ordering amongst the roots. Thus if ${\ca Z}_p = \{\lambda_1, \dots, \lambda_n\}$ with each $\lambda_j \in \bb C$, $\sigma {\cal Z}_p = \{\lambda_{\sigma(1)}, \dots, \lambda_{\sigma(n)}\}$ denotes the same set of roots for every permutation $\sigma$ in the permutation group $S_n$. Hence ${\cal Z}_p$ is more naturally viewed as an element of the quotient space $\bb C^n/S_n$, where the underlying equivalence relation $u \sim v$ is via,
\begin{align*}
  u = (u_1, \dots, u_n)^\top = (v_{\sigma(1)}, \dots, v_{\sigma(n)})^\top,
\end{align*}
for some $\sigma \in S_n$;
endow this quotient space $\bb C^n/S_n$ with a quotient topology induced by the canonical projection $\pi:~\bb C^n \to \bb C^n/S_n$. \par
The following result will subsequently be used in our analysis.
%states that there is a homeomorphism between the set of coefficients of a monic complex polynomial and its zeros.
%\begin{theorem}{(Theorem $A$ in~\cite{harris1987shorter})}
\begin{theorem}(\cite{harris1987shorter})
  \label{thm:zeros}
  There is a homeomorphism 
   $h \colon \bb C^n \to \bb C^n/S_n,$
mapping the coefficients of a monic complex polynomial to its zeros.
\end{theorem}
In this manuscript, we are concerned with ``real'' linear time-invariant (LTI) systems, i.e., systems with real parameters and feedback gains. Hence the roots of the corresponding characteristic polynomial
$p(x) \in \bb R[x]$ will be invariant under complex conjugation, namely if ${\cal Z}_p = \{z_1, \dots, z_n\} \in \bb C^n/S_n$, then\footnote{Note the equivalence relation on ${\bb C}^n/S_n$.}
\begin{align*}
  \bar{\cal Z}_p = \{\bar{z}_1, \dots, \bar{z}_n\} = {\cal Z}_p;
\end{align*}
denote by $\bb C^n_{*}$ as the set of  vectors in $\bb C^n$ that are invariant under entry-wise conjugation.

The relation between the coefficients and roots for real polynomials follows from Theorem~\ref{thm:zeros}, as the restriction of a homeomorphism is again a homeomorphism.
\begin{corollary}
  \label{cor:real_zeros}
  Suppose that $p(\lambda)= \lambda^n + a_{n-1}\lambda^{n-1} + \dots + a_0 \in \bb R[\lambda]$ is a real polynomial. Then there is a homeomorphism
    $\hat{h}: \bb R^n \to \bb C^n_{*}/S_n$, mapping coefficients of $p(\lambda)$ to its roots.
\end{corollary}
Consider now the continuous LTI SISO system,
\begin{equation}\label{eq:continuous_system}
\begin{split}
  \dot{\bf x}(t)= A {\bf x}(t) + b {\bf u}(t), \quad  {\bf y}(t) = c^\top {\bf x}(t),
\end{split}
\end{equation}
and its discrete time variant, 
\begin{equation}\label{eq:discrete_system}
\begin{split}
 {\bf x}(k+1)= A {\bf x}(k) + b {\bf u}(k), \quad {\bf y}(k)& = c^\top {\bf x}(k),
\end{split}
\end{equation}
where $A \in {\bb M}_n(\bb R)$ and $b, c \in \bb R^n$; such systems will be abbreviated in terms of the 
triplet $(A, b, c^\top)$. We say that the system $(A, b, c^\top)$ is controllable and observable if it satisfies the Kalman Rank Condition~\cite{zabczyk2009mathematical}, namely, $\text{\bf rank}([b, Ab, \dots, A^{n-1}b])=n$ and $\text{\bf rank}([c^\top, c^\top A, \dots, c^\top A^{n-1}]^\top) = n$.
 For synthesizing state feedback gains for SISO systems with dynamics of order $n$, we are interested in identifying $k \in \bb R^n$ to synthesize the control signal ${\bf u}(t) = k^\top {\bf x}(t)$; if output-feedback is of interest, the feedback gain is a scalar $k \in \bb R$ and ${\bf u}(t) = k {\bf y}(t) = k c^\top {\bf x}(t)$. For a controllable and observable triplet $(A, b, c^\top)$, we denote the set of \emph{Hurwitz stabilizing output feedback gains} as,\footnote{The $\max$ operation on a set identifies the maximum element of that set.}
\begin{align}
  \label{eq:continuous_H}
  \ca H &= \{k \in \bb R: \max \{ \text{\bf Re}({\cal Z}_{p(\lambda, {A-kbc^\top})})\} < 0\},
          \end{align}
and the set of \emph{Schur stabilizing output-feedback gains} by
\begin{align}
  \label{eq:discrete_S}
  \ca S &=\{k \in \bb R: \max \{ |{\cal Z}_{p(\lambda, {A-kbc^\top})}|\} < 1\},
\end{align}
where $p(\lambda, A-kbc^\top)$ denotes the characteristic polynomial of the closed-loop system with feedback gain $k \in \bb R$.
% \text{\bf Re} designates the real part of the complex number (or a set of them), modulus of a set is the set of modulus' of each of its elements, and $\max$ when applied to a set of real numbers returns its maximum element.
Naturally, we could have defined the sets $\ca H$ and $\ca S$ in terms of the eigenvalues of $A-kbc^\top$;
thus $\max \{ |{\cal Z}_{p(\lambda, {A})}|\}$ is simply the spectral radius of the matrix $A$, that we denote by $\rho (A)$.

When we examine state feedback with the same system parameters $(A, b)$, the sets $\ca H_{\bf x}$ and $\ca S_{\bf x}$ are defined as,
\begin{align}
  \label{eq:continuous_H_s}
  \ca H_{\bf x} &= \{k \in \bb R^n: \max \{ \text{\bf Re}({\cal Z}_{p(x, {A-bk^\top})}) \} < 0\},\\
  \label{eq:discrete_S_s}
  \nonumber \ca S_{\bf x} &= \{k \in \bb R^n: \max \{ |{\cal Z}_{p(\lambda, {A-kbc^\top})}|\} < 1\}\\
  &= \{k \in \bb R^n: \rho(A-kbc^\top) < 1\},
\end{align}
where we have used the subscript $\bf x$ to denote state feedback.\footnote{That is, the state ``${\bf x}$'' is available for feedback.}

We now observe a relation between $\ca H$ and $\ca H_{\bf x}$; a similar relation holds between $\ca S$ and $\ca S_{\bf x}$. This relation will be used in our subsequent analysis.
\begin{observation}
  \label{obs:state-output-relation}
  For a controllable and observable system $(A, b, c^\top)$,
  \begin{align*}
    \ca H = \{ k \in \bb R: (kc) \cap \ca H_{\bf x} \neq \emptyset\}.
  \end{align*}
\end{observation}
\begin{proof}
  We only need to observe that $k \in \ca H$ is equivalent to having $k c \in \ca H_{\bf x}$.
\end{proof}
Denote by $(A^{\flat}, b^{\flat})$ the {controllable canonical form}~\cite{kailath_Book_1980} of the pair $(A, b)$;
$\ca H^{\flat}_{\bf x}, \ca S^{\flat}_{\bf x}$ then denote the corresponding set of Hurwitz/Schur stabilizing state feedback gains~\cite{kailath_Book_1980}. {We now observe that the sets $\ca H^{\flat}_{\bf x}$ and $\ca S^{\flat}_{\bf x}$ are related to $\ca H_{\bf x}$ and $\ca S_{\bf x}$ through a change of coordinates.
% of $\ca H_{\bf x}, \ca S_{\bf x}$.}
%
%We now observe that the set of stabilizing feedback controllers for $(A, b)$ is diffeomorphic to the set of stabilizing feedback gains of $(A^{\flat}, b^{\flat})$.
\begin{observation}
  \label{obs:feedback_equiv}
  Let $(A, b)$ be a controllable pair and $(A^{\flat}, b^{\flat})$ be its corresponding {controllable canonical form}. Then % $\ca S^{\flat}_{\bf x}$ and $\ca H^{\flat}_{\bf x}$ are diffeomorphic to $\ca S_{\bf x}$ and $\ca H_{\bf x}$, respectively.
{$\ca S_{\bf x} = T^\top\ca S^{\flat}_{\bf x} \coloneqq \{T^\top k: k \in \ca S^{\flat}_{\bf x}\}$ and $\ca H_{\bf x}=T^\top\ca H^{\flat}_{\bf x} \coloneqq \{T^\top k: k \in \ca H^{\flat}_{\bf x}\}$, where $T \in \bb{GL}_n(\bb R)$ and $(A^{\flat}, b^{\flat}) = (TAT^{-1}, Tb)$.}
\end{observation}
\begin{proof}
  We only prove % $\ca S^{\flat}_{\bf x}$ is diffeomorphic to $\ca S_{\bf x}$;
  {the relation between $\ca S^{\flat}_{\bf x}$ and $\ca S_{\bf x}$;}
  the proof for $\ca H^{\flat}_{\bf x}$ and $\ca H_{\bf x}$ follows analogously.
%  \sout{By our assumption, there is $T \in \bb{GL}_n(\bb R)$ such that}
%\[ \cancel{A^{\flat} = T A T^{-1}, \quad b^{\flat} = Tb.} \]
%
When $k \in \ca S^{\flat}_{\bf x}$, 
% i.e., $\rho(A-bk^\top) < 1$,
\begin{align*}
  \rho(A^{\flat}-b^{\flat}k^\top) = \rho(T(A - bk^\top T)T^{-1}) = \rho(A-bk^\top T) < 1;
\end{align*}
hence, $T^\top k \in \ca S_{\bf x}$ and $T^\top\ca S^{\flat}_{\bf x} \subseteq \ca S_{\bf x}$. The set inclusion in the other direction follows analogously; as such, $T^\top\ca S^{\flat}_{\bf x} = \ca S_{\bf x}$.
%\coloneqq \{T^\top k: k \in \ca S^{\flat}_{\bf x}\}$. 
%\sout{Since the map $k \mapsto T^\top k$ is a diffeomorphism on $\bb R^n$, we conclude $\ca S_{\bf x}$ and $\ca S^{\flat}_{\bf x}$ are diffeomorphic.}
\end{proof}
\begin{remark}
  {Since $k \mapsto T^{\top}k$ is a diffeomorphism on $\bb R^n$, topological properties of the set of stabilizing feedback gains, such as connectedness, can be studied for the {controllable canonical form instead.}} 
%  \sout{This observation shows that topological properties such as connectedness, can be studied for the {\color{blue} controllable canonical form} instead. This is also valid for a controllable and observable triple $(A, b, c^\top)$, i.e., we may without loss of generality assume $(A, b)$ is in the {\color{blue} controllable canonical form}. This observation is useful as the characteristic polynomial of the closed-loop system for $(A^{\flat}, b^{\flat}, c^\top)$ admits a rather simple form. Namely, if $\mathbf r_n(A^{\flat}) = (a_0, a_1, \dots, a_n)$, where $\mathbf r_j(A)$ denotes the $j^{\text{th}}$ row of $A$, the characteristic function of the closed-loop system for a feedback gain $k \in \bb R$ is given by}
%\[  \cancel{p(\lambda, A-kbc^\top) = \lambda^n + (a_{n-1}-k c_{n-1}) \lambda^{n-1} + \dots + (a_0 - k c_0). }\]
{Furthermore, for a controllable and observable triplet $(A, b, c^\top)$, we may assume that $(A, b)$ is in the controllable canonical form in order to characterize the stabilizing output feedback gains. In fact, since $(A^{\flat}, b^{\flat}) = (TAT^{-1}, Tb)$, the set of stabilizing gains for $(TAT^{-1}, Tb, c^\top  T^{-1})$ will coincide with $\ca H$ and $\ca S$ for the triple $(A, b, c^\top)$. This observation is useful as the characteristic polynomial of the closed-loop system for $(A^{\flat}, b^{\flat}, c^\top)$ admits a rather simple form. 
%Namely, if $\mathbf r_n(A^{\flat}) = (a_0, a_1, \dots, a_n)$, where $\mathbf r_j(A)$ denotes the $j^{\text{th}}$ row of $A$, the characteristic polynomial of the closed-loop system for a feedback gain $k \in \bb R$ is given by
%\[  p(\lambda, A-kbc^\top) = \lambda^n + (a_{n-1}-k c_{n-1}) \lambda^{n-1} + \dots + (a_0 - k c_0). \]
}
\end{remark}
\section{Properties of Hurwitz Stabilizing Feedback Gains}
\label{sec:hurwitz}
Consider again the continuous LTI SISO system~(\ref{eq:continuous_system})
in relation to the sets $\ca H$~\eqref{eq:continuous_H} and $\ca H_{\bf x}$~\eqref{eq:continuous_H_s}.
% In this section, we concern the continuous linear time-invariant system
% \begin{align*}
  % &\dot{x}(t) = Ax(t) + b u(t), \\
  % & y(t) = c^T x(t),
% \end{align*}
% where $A \in M_n(\bb R)$, $b, c \in \bb R^n$ and $u(t)$ is a scalar for every $t \in \bb R$. If state-feedback system is of interest, the dynamic is
% \begin{align*}
  % &\dot{x}(t)=Ax(t)+b \tilde{u}(t),
% \end{align*}
% where $A \in M_n(\bb R)$, $b \in \bb R^n$ and $\tilde{u}(t) \in \bb R^n$ for every $t \in \bb R$. \par
The diffeomorphism between the set of stabilizing feedback gains for a system and its {controllable canonical form} (Observation~\ref{obs:feedback_equiv}) allows us
to prove a number of topological and metrical properties for $\ca H$ and $\ca H_{\bf x}$.
In fact, Observation~\ref{obs:feedback_equiv} leads to the following properties through 
the application of theory of polynomials: 
(a) $\ca H$ is open in the Euclidean topology for both state-feedback and output-feedback systems,
(b) $\ca H_{\bf x}$ is unbounded but $\ca H$ can be either bounded or unbounded,
(c) the sets $\ca H$ and $\ca H_{\bf x}$ are both convex when the corresponding system has order two,
(d) $\ca H_{\bf x}$ is connected and $\ca H$ can have at most $\ceil{\frac{n}{2}}$ connected components. 
We now provide the proofs for these observations.
\begin{lemma}
  \label{lemma:open_cont}
 The set $\ca H$ is open in $\bb R$ and $\ca H_{\bf x}$ is open in $\bb R^n$.
\end{lemma}
\begin{proof}
  By Observation~\ref{obs:feedback_equiv}, without loss of generality, we shall assume that the system
$(A, b, c^\top)$ is in {the controllable canonical form}. 
  Let $a = (a_0, \dots, a_{n-1})$ be the last row of $A$ and $c = (c_0, \dots, c_{n-1})^\top$. 
  For any $k \in \bb R$, the characteristic polynomial of this system is given by,
  \begin{align*}
    p(\lambda, k) = \lambda^n + (a_{n-1}-k c_{n-1}) \lambda^{n-1} + \dots + (a_0 - k c_0).
  \end{align*}
  We note that for a fixed $k \in \bb R$,
  the map $\tilde{\upsilon} : \mathbb C^n_{*}/S_n \to \bb R$ given by
  \begin{align*}
    {\cal Z}_{p(\lambda,k)} \mapsto \max \{ \text{\bf Re}({\cal Z}_{p(\lambda, k)}) \},
  \end{align*}
  is continuous since $\upsilon \coloneq \max \circ \, \text{\bf Re}: \mathbb C^n_{*} \to \bb R$ is continuous and there is a unique continuous map $\tilde{\upsilon}: \bb C^n_{\*}/S_n \to \bb R$ such that $\upsilon = \tilde{\upsilon} \circ \pi$:
\begin{figure}[H]
        \centering
       \begin{tikzcd}
          \bb C^n \arrow[d, "\pi"] \arrow[rd, "\upsilon"] & \\
          \bb C^n_{\text{*}}/S_n \arrow[r, "\tilde{\upsilon}"] & \bb R;
       \end{tikzcd}
      \end{figure}
this follows from the properties of the quotient topology (see Theorem 3.73 in~\cite{lee2010introduction}).
Thus the map,
  \begin{align*}
    g: k \mapsto {\cal Z}_{p(\lambda,k)} \mapsto \max \{ \text{\bf Re}({\cal Z}_{p(\lambda,k)})\},
  \end{align*}
  is continuous as it is a composition of continuous maps. 
  Thus as the pre-image of the open interval $(-\infty, 0)$ under the continuous map $g$, the set
  $\ca H$ is an open subset of $\bb R$ and as such, $\ca H$ is a union of disjoint open 
  intervals.\footnote{That is, it can always be represented as such.}
  %  \par
 % Recall that the set $\ca H_{\bf x}$ is a subset of $\bb R^n$. 
  Following a similar argument, the map $g_{\bf x}: \bb R^n \to \bb R$ defined via the composition,
    \begin{align*}
    g_{\bf x} : k \mapsto {\cal Z}_{p(\lambda,k)} \mapsto \max \{ \text{\bf Re}({\cal Z}_{p(\lambda,k)}) \}, 
  \end{align*}
  is continuous and 
  %$\ca H_{\bf x} = g_{\bf x}^{-1}((-\infty, 0))$. 
  hence $\ca H_{\bf x}$ is open in $\bb R^n$.
\end{proof}
We shall point out another favorable property of $\ca H_{\bf x}$, namely that it is {\em regular} open. In other words, the closure of the set of Hurwitz stabilizing controllers is the set of marginally stabilizing controllers and $\ca H_{\bf x}$ is precisely the interior of the set of marginally stabilizing controllers.
\begin{lemma}
  \label{lemma:regular_open}
  Let $$\ca B_{\bf x} = \{k \in \bb R^n: \max \{ \text{\bf Re}({\cal Z}_{p(\lambda, A-bk^\top)}) \} = 0\}.$$
  Then $\ca B_{\bf x}$ is the boundary of $\ca H_{\bf x}$ (that is, $\partial \ca H_{\bf x}$) and the closure
  of ${\ca H}$ is,
  \begin{equation}
 \bar{\ca H}_{\bf x} = \ca H_{\bf x} \cup \ca B_{\bf x} = \{k \in \bb R^n: \max \{ \text{\bf Re}({\cal Z}_{p(\lambda, A-bk^\top)}) \} \le 0\}. \label{closure-of-H}
  \end{equation}
 \end{lemma}
\begin{proof}
  Note that for any $k \in \ca B_{\bf x}$, we can perturb the marginally stable eigenvalues to be stable, i.e., having negative real parts. This means that there exists a sequence $\{k_n\} \subseteq \ca H_{\bf x}$ such that $k_n \to k$. On the other hand, we may as well perturb the marginally stable eigenvalues to become unstable, i.e., having positive real parts, suggesting that there is a sequence $\{k_n'\} \subseteq \ca H_{\bf x}^c$ such that $k_n' \to k$. Hence $\partial H_{\bf x} = \ca B_{\bf x}$ and (\ref{closure-of-H}) follows.
 \end{proof}
It is also immediate to deduce that the interior of $\bar{\ca H}_{\bf x}$,
namely $\left(\bar{\ca H}_{\bf x}\right)^{\circ},$ is $\ca H_{\bf x}$; thus $\ca H_{\bf x}$
is regular open.

Let us now examine the boundedness of the sets $\ca H$ and $\ca H_{\bf x}$.
%By Pole Shifting Theorem~\cite{rissanen1960control}, it is rather clear that $\ca H_s$ is unbounded.
\begin{observation}
The set $\ca H_{\bf x}$ is unbounded. 
\end{observation}
\begin{proof}
  By Observation~\ref{obs:feedback_equiv}, it suffices to assume that the pair $(A, b)$ is in {controllable canonical form}. For any $k = (k_0, \dots, k_n)^\top \in \bb R^n$, the characteristic polynomial of the 
  corresponding closed-loop system is,
\begin{align}
  p(\lambda,k) = \lambda^n + (a_{n-1}-k_{n-1}) \lambda^{n-1} + \dots + (a_0 - k_0). \label{poly-k}
\end{align}
  By Pole-shifting theorem~\cite{rissanen1960control}, for every $n$-tuple $(-j, -j, \dots, -j)$, with $j \in \bb R$, there is some $k^j \in \bb R^n$ such that $k^j \in \ca H_{\bf x}$ and the zeros of $p(\lambda, A-b(k^j)^\top)$ are exactly $(-j, \dots, -j)$. But $a_0 - k_0^j = (-j)^n$ by Vieta's formula.\footnote{The indexing of  $k^j$ is with reference to (\ref{poly-k}).} Hence $\ca H_{\bf x}$ is not bounded.
  \end{proof}
For an output feedback system, the set $\ca H$ can either be bounded or unbounded depending on the properties of the system $(A, b, c^\top)$; this is demonstrated in the following example.
\begin{example}
  Let the triplet $(A^{\flat}, b^{\flat}, c^\top)$, in controllable canonical form, be controllable.\footnote{Of course, the controllability of the triplet is independent of $c$.} Let $a = (a_{0}, a_{1}, \dots, a_{n-1})$ be the last row of $A^{\flat}$. Then,
  \begin{enumerate}
    \item If for some $i, j \in \{0, \dots, n-1\}$, $c_i > 0$ and $c_j < 0$, then $\ca H$ is bounded.
    \item Suppose that $n=4$ and the entries of $c$ are positive. If $c_3 c_2 c_1 < c_3^2 c_0$ then $\ca H$ is unbounded; when $c_3 c_2 c_1 > c_3^2 c_0$, $\ca H$ is bounded. 
  \end{enumerate}
The assertions in this example are consequences of the Routh-Hurwitz Criterion.
If $k$ is stabilizing, then 
\begin{align*}
  a_{n-1} - k c_{n-1} > 0, \dots, a_0 - k c_0 > 0.
\end{align*}
If there exists a $k$ that satisfies the above inequalities, and for some $i,j$, 
$c_i > 0$ and $c_j < 0$, then {$k$ satisfies,} 
\begin{align*}
  \frac{a_j}{c_j} < k < \frac{a_i}{c_i}.
\end{align*}
Hence, $\ca H$ must be bounded. For part $(b)$, 
the Routh-Hurwitz Criterion states that,
\begin{align*}
  &a_3 - k c_3 > 0, a_2 - k c_2 > 0, a_1 - kc_1 >0, a_0 - k c_0 > 0, \\
  &(a_3-kc_3)(a_2-kc_2)(a_1 - kc_1) \\
  &\quad > (a_1 -kc_1)^2 + (a_3-kc_3)^2(a_0-kc_0).
\end{align*}
%Hence, $k < \alpha$ for some $\alpha \in \bb R$. 
We note that for sufficiently negative $k$, the last inequality holds since $c_1c_2c_3 >  c_3^2 c_0$; hence $\ca H$ is not bounded. On the other hand, if $c_1c_2c_3 < c_3^2 c_0$, then $k$ must be bounded from below; thus $\ca H$ is bounded. 
\end{example}
We further make an observation on a necessary condition for the unboundedness of the set $\ca H$; see also~\cite{feng2019}.
\begin{observation}
  \label{obs:condition_unbounded}
If $(A, b, c^\top)$ is controllable and observable, a necessary condition for $\ca H$ to be unbounded is 
that the nonzero entries of $c$ have the same sign. Moreover, if $\ca H$ is unbounded, then $\ca H$ must only include one of the two intervals: $(-\infty, M)$ and $(M', \infty)$ for some $M, M' \in \bb R$.
\end{observation}
\begin{proof}
  If a monic polynomial is Hurwitz stable, all of its coefficients are positive (Theorem $2.4$ in~\cite{zabczyk2009mathematical}). Hence, we have $a_j - k c_j > 0$ for every $j$. If the nonzero entries of $c$ do not have the same sign, $k$ should be bounded. Since the system is observable, $c \neq 0$. 
Thereby, either $k < M$ or $k > M'$, for some $M, M' \in \bb R$.
\end{proof}
We now make a few observations on the convexity of the sets $\ca H$ and $\ca H_{\bf x}$;
needless to say, these observations have direct algorithmic implications. 
It is known that a convex combination of stable polynomials is not necessarily convex~\cite{bialas1985convex}. However, a convex combination of stable monic polynomials with degree $2$ is convex; this follows from the Routh-Hurwitz criterion. 
%In fact, we observe that the sets $\ca H$ and $\ca H_{\bf x}$ are convex if the system has two states, i.e., $n=2$.
\begin{observation}
  \label{obs:state_convex}
  For the state feedback system $(A,b)$, the set
  $\ca H_{\bf x}$ is convex if $n=2$.
\end{observation}
\begin{proof}
  It suffices to show this for $(A^{\flat}, b^{\flat})$. Let $k = (k_0, k_1)^\top$ and $k' = (k'_0, k'_1)^\top$ be two stabilizing feedback gains. Then the characteristic polynomial of the corresponding closed-loop systems are,
  \begin{align*}
    p_1(\lambda) = \lambda^2 + (a_1 - k_1) \lambda + (a_0 - k_0), \\
    p_2(\lambda) =\lambda^2 + (a_1 - k'_1) \lambda + (a_0 - k'_0). 
  \end{align*}
Note that $p_1$ and $p_2$ are stable if and only if the coefficients are positive.\footnote{For convexity analysis, this is in fact the key property for systems of order $2$.} Hence,
if ${\hat k} = (1-\delta)k + \delta k'$, for $\delta  \in (0,1),$ 
then $p_{\hat k}(x) = (1-\delta) p_1(x) + \delta p_2(x)$ and $p_{\hat k}$ is stable by the positivity of its coefficients.
\end{proof}
\begin{observation}
  For the output feedback system $(A, b, c^\top)$ with $n=2$, the set $\ca H$ is convex.
\end{observation}
\begin{proof}
  Recall that by Observation~\ref{obs:state-output-relation},
  \begin{align*}
    \ca H = \{ k \in \bb R: (kc) \cap \ca H_{\bf x} \neq \emptyset\}.
  \end{align*}
  Noting that $\text{\bf span}(c) \cap H_{\bf x}$ is a convex subset of $\bb R^2$; by Observation~\ref{obs:state_convex}, the proof follows.
\end{proof}
We can also use the Routh-Hurwitz stability criteria to show the nonconvexity of
the set of stabilizing feedback gains when $n > 2$. For example, let $n=3$ and consider the {controllable canonical form} $(A^{\flat}, b^{\flat})$ with the last row of $A^{\flat}$ set to $0$. Then the stabilizing feedback gain is parametrized by three parameters $k = (k_1, k_2, k_3)$ and the characteristic polynomial of $A-bk^\top$ is 
\begin{align*}
  p(\lambda, A-bk^\top) = \lambda^3 + k_3 \lambda^{2}+ k_2 \lambda + k_1.
\end{align*}
The following example essentially shows that convex combinations of stable polynomials are not necessarily stable.%
%{\color{blue}
\begin{example}
  Consider the system, 
\begin{align*}
  A = 
        \begin{pmatrix}
          0 & 1 & 0 \\
          0 & 0 & 1 \\
          0 & 0 & 0
        \end{pmatrix}, \; b=\begin{pmatrix}
          0 \\
          0 \\
          1
        \end{pmatrix}.
\end{align*}
For $k_1 = (-24, -5, -5)^{\top}$ and $k_2 = (-0.9, -1, -1)^{\top}$, the characteristic
  polynomials of the corresponding closed-loop systems are given by
%\begin{align*}
%  p_1(\lambda) = \lambda^3 + \lambda^2 + 10 \lambda + 9, \; \mbox{and} \; 
%  p_2(\lambda) = \lambda^3 + \lambda^2 + \lambda + 0.9.
%\end{align*}
$  p_1(\lambda) = \lambda^3 +5 \lambda^2 + 5 \lambda + 24$ and $p_2(\lambda) = \lambda^3 + \lambda^2 + \lambda + 0.9.$
Both polynomials are stable and as such, $k_1, k_2 \in \ca H_{\bf x}$; however, $k' = 0.5 k_1 + 0.5 k_2$ yields an unstable characteristic polynomial $p(\lambda)
%= 0.5 p_1(\lambda)+0.5 p_2(\lambda) 
= \lambda^3 +3  \lambda^2 + 3 \lambda + 12.45.$ 
%for the closed loop system.
\end{example}
%----
\subsection{Connectedness properties of $\ca H_{\bf x}$ and $\ca H$}
We will now delve into connectedness of the 
sets $\ca H$ and $\ca H_{\bf x}$, requiring more delicate arguments as compared with
their boundedness and convexity properties.
For the state feedback system,
\begin{align*}
  \dot{\bf x}(t) = A {\bf x}(t) + b{\bf u}(t),
\end{align*}
by Corollary~\ref{cor:real_zeros}, $\ca H_{\bf x}$ is connected in $\bb R^n$. 
We now show that this set is in fact contractible, i.e., it can be continuously deformed to a point~\cite{lee2010introduction}.
\begin{lemma}
  \label{lemma:continuous_state_connected}
When $(A, b)$ is controllable, the set of stabilizing feedback controllers $\ca H_{\bf x} \subseteq \bb R^n$
  is connected and contractible.
\end{lemma}
\begin{proof}
  Let $\{\bb H_{-}^n\}_{*}$ denote the set of $n$-tuples $v \in \bb H_{-}^{n}$ invariant under (entry-wise) complex conjugation. We note that  $\{\bb H_{-}^n\}_{*}/S_n$ is connected in $\bb C^n/S_n$ by noting that every $v \in \{\bb H_{-}^n\}_{*}/S_n$ is path-connected.
  %\sout{and indeed contractible} to $(-1, -1, \dots, -1)^\top \in \bb C^n$ via $t \mapsto (1-t)v + t(-1, \dots, -1)^\top$. % 
{This path in fact defines a homotopy between the identity and the constant map $(-1, \dots, -1)$.}
  %
%  \sout{Hence, $\pi(\{\bb H_{-}^n\}_{*})$ is connected in $\bb C^n_{*}/S_n$ (recall that $\pi$ is the canonical projection $\pi: \bb C^n \to \bb C^n/S_n$; see \S~\ref{sec:notations}.)} 
  Mind that $$\ca H = (a_{n-1}, \dots, a_0)^\top - \hat{\sigma}^{-1}(\{\bb H_{-}^n\}_{*}/S_n),$$ i.e., an affine translation in $\bb R^n$. By Corollary~\ref{cor:real_zeros}, it now follows that $\ca H_{\bf x}$ is connected and contractible.
\end{proof}
\begin{remark}
  Note that even though the set $\ca H_{\bf x}$ in contractible to a point, it is not necessary star-convex. This is due to nonlinearity of the Vieta's map $\hat{\sigma}^{-1}$.
\end{remark}
%The set $\ca H$ of stabilizing output feedback gains 
%is not connected in general. As we show next, however, 
%its number of its connected components
%is bounded by $\ceil{\frac{n}{2}},$ where $n$ is the number of states. 
%-------
\subsubsection*{\bf Connected Components of $\ca H$}
\label{subsec:comonents_continuous}
 We now develop bounds on the number of connected components of $\ca H$ as it is not necessary connected. Lemma~\ref{lemma:2n_bound_continuous} provides a bound of $n$ and Lemma~\ref{lemma:n_bound_continuous} will tighten the bound to $\ceil{\frac{n}{2}}$. 
   We have chosen to present the two lemmas in sequence, 
since the proof of Lemma~\ref{lemma:2n_bound_continuous} is straightforward but tightening the bound to $\ceil{\frac{n}{2}}$ requires more delicate analysis.
%\begin{remark}
%We note that upper bounds of $n$ and $\ceil{\frac{n}{2}}$ for the number of connected components of the set of stabilizing output gains for SISO case have recently been reported, in~\cite{feng2018on}, \cite{feng2019}, respectively.
%%
%However, the proposed lines of reasoning in \cite{feng2018on,feng2019 }, do not take into account the possibility of having adjacent stabilizing connected components (see Remark~\ref{remark:continuous_tangent} and Figure~\ref{fig:continuous_tangent} to demonstrate this phenomena). 
%%
%%  Note a bound of $n$ appears in the unpublished manuscript~\cite{feng2018on} (Theorem $1$). However, the proof there is not correct (see Remark~\ref{remark:continuous_tangent} and Figure~\ref{fig:continuous_tangent} for details). 
%\end{remark}

Let us start our analysis by recalling the smooth dependence of simple roots
of a polynomial on its coefficients, subsequently used in
Lemma~\ref{lemma:n_bound_continuous}.
\begin{lemma}
  \label{lemma:smooth_dependence}
  Let $a_j: {\cal I} \to \bb R$ be $C^{\infty}$ functions for $j = \{1, \dots, n\}$, where ${\cal I} \subseteq \bb R$ is an open interval. If $t_0 \in {\cal I}$ and $\lambda_0$ is a simple root of the polynomial $f(\lambda) = a_n(t_0) \lambda^n + a_{n-1}(t_0) \lambda^{n-1} + \dots + a_0(t_0) \in \bb R[\lambda]$ with $t_0 \in {\cal I}$, then there exists a $C^{\infty}$ function $\eta: {\cal J} \to \bb C$ over an open interval ${\cal J} \subseteq {\cal I}$ such that $t_0 \in {\cal J}$, $\eta(t_0) = \lambda_0$ and $\eta(t)$ is a zero of $$f(\lambda,t)= a_n(t) \lambda^n + a_{n-1}(t) \lambda^{n-1} + \dots + a_0(t),$$ for every $t \in {\cal J}$.
\end{lemma}
\begin{proof}
  This follows from Implicit Function Theorem~\cite{krantz2012implicit}. First note that $f(\lambda, t)$ is $C^{\infty}$ in both $\lambda$ and $t$; we note that at $(\lambda_0, t_0)$, $f'(\lambda_0, t_0) \neq 0$ since $\lambda_0$ is a simple zero. 
\end{proof}
We now prove an upper bound of $n$ on the number of connected components of $\ca H$.
\begin{lemma}
  \label{lemma:2n_bound_continuous}
  If $\ca H \neq \emptyset$, it has at most $n$ connected components.
\end{lemma}
\begin{proof}
  Note that in the SISO case, $\ca H$ is a subset of $\bb R$; furthermore, 
it suffices to show that $(A^{\flat}, b^{\flat}, c^\top)$ has at most $n$ connected components. 
 Recall that for $k \in \bb R$, the characteristic polynomial of a closed-loop system $A^{\flat}-k b^{\flat} c^\top$ is given by
  \begin{align*}
    p(\lambda, k) = \lambda^n + (a_{n-1} - k c_{n-1}) \lambda^{n-1} + \dots + a_0 - k c_0,
  \end{align*} where $a \coloneqq (a_0, a_1, \dots, a_{n-1})$ is the last row of $A^{\flat}$ and $c_j$'s are components of $c$.
  Let $\zeta: \mathbb R^n \to \ca P_n(\lambda)$ denote the natural bijection, assigning 
  coefficients to monic polynomials.  We denote by,
  \begin{align*}
    \Gamma = \{a \in \bb R^n: \zeta(a) \text{ has at least one zero on imaginary axis}\},
  \end{align*}
  and $\ell(k) = a - k c$, i.e., a parametrized line in $\bb R^n$.
  Suppose that $\ell (k)$ intersects $\Gamma$ for finitely many $k$'s, listed in an increasing order $k_1, \dots, k_q$; this fact will be proved subsequently. Let $n_{p(\lambda, k)}(\bb H_{-})$ denote the number of roots of $p(\lambda, k)$ in $\bb H_{-}$. Moreover, let $\gamma(r)$ be a counterclockwise oriented curve in $\bb C$ consisting the line segment $[-ir, ir]$ and the semicircle $S(r,\theta)= r e^{i \theta}$, with $\theta \in [\pi/2, 3\pi/2]$. For each $k \in (k_j, k_{j+1})$, we define 
\begin{align*}
m_r(k)= \frac{1}{2\pi i}\int_{\gamma(r)} \frac{p'(\lambda,k)}{p(\lambda, k)} \, d\lambda, \; \mbox{and} \; m(k) = \lim_{r \to \infty} m_r(k).
\end{align*}
Note that if $p(\lambda, k)$ does not vanish on $\gamma(r)$, by Cauchy's Argument Principle~\cite{rudin2006real}, 
$m_r(k)$ 
%$\int_{\gamma(r)} \frac{p'(z,k)}{p(z,k)}dz$ 
is the number of zeros of $p(\lambda, k)$ inside the curve $\gamma(r)$. 
However, since $p(\lambda, k)$ has at most $n$ roots, the integral is well-defined except at finitely many $r$'s. 
Hence $m(k)$ is well-defined and $m(k) = n_{p(\lambda,k)}(\bb H_{-})$. 
In the meantime, the function $m(k)$ is continuous in $k$ and integer-valued, and thereby, 
$m(k) = n_{p(\lambda, k)}$ is constant on each interval $(k_j, k_{j+1})$. 
That is, either $n_{p(\lambda, k)} = n$ or $n_{p(\lambda, k)} < n$, corresponding to either stabilizing or non-stabilizing gains, respectively. So by inspecting the number of intersections between $\ell(k)$ and $\Gamma$, one can derive an upper bound on the number of connected components of $\ca H$.
Let 
  \begin{equation}
r(\lambda) = c_{n-1}\lambda^{n-1} + \dots + c_0, \; \mbox{and} \; 
%  \end{equation}
%and 
%\begin{equation}
s(\lambda) = a_{n-1} \lambda^{n-1} + \dots + a_0. \label{rs}
\end{equation}
Consider an intersection of $\ell(k)$ and $\Gamma$ that is, when for 
some $k \in \bb R$, there exists $\lambda = i \beta$, $\beta \in \bb R$, 
for which $p(i\beta, k) =0$. 
We first observe that $r(i\beta) \neq 0$ since otherwise $i\beta$ would be a root 
for every $k \in \bb R$. This on the other hand, implies that $\ca H$ is empty. 
Hence $p(i\beta, k) = 0$ implies that,
  \begin{align}
    k = \frac{(i\beta)^n + a_{n-1} (i\beta)^{n-1} + \dots + a_0}{c_{n-1} (i\beta)^{n-1} + \dots + c_0}. \label{eq:myeq1_c}
  \end{align}
  Since $k \in \bb R$, $\beta$ must be a root of,
  \begin{align*}
    \phi(\beta) = \text{\bf Im}\left(\left(\lambda^n + s(\lambda)\right)r(\bar{\lambda})|_{\lambda=i \beta}\right) = 0.
  \end{align*}
  We note that $\phi(\beta) \in \bb R[\beta]$ has degree at most $2n-1$. Thus it can be written as,
  \begin{align*}
    \phi(\beta) = i^{2n-2} \beta^{2n-1} + i^{2n-4} d_{2n-3} \beta^{2n-3} + \dots + i^2 d_3 \beta^3 + d_1 \beta,
  \end{align*}
  for some set of real coefficients $\{d_{2n-3}, d_{2n-5}, \dots, d_1\} \subseteq \bb R$,
noting that all exponentials of $i$ are either $1$ or $-1$. 
Now let us set $\tau(\beta) = i^{2n-2} \beta^{2n-2} + i^{2n-4} d_{2n-3} \beta^{2n-4} + \dots + i^2 d_3 \beta^2 + d_1 \in \bb R[\beta]$; hence, $\phi(\beta) = \beta\tau(\beta)$. 
Now we note that $\tau(\beta)$ is an even polynomial in $\beta$ (having only even degrees).
% terms and we are interested in real solutions. Putting
Letting
  \begin{align}
\label{eq:alg_poly_cont}
    \upsilon(\beta) = i^{2n-2} \beta^{n-1} + i^{2n-4} d_{2n-3} \beta^{n-2} + \dots + i^2 d_3 \beta + d_1,
  \end{align}
we have  $\tau(\beta) = \upsilon(\beta^2)$. 
Thereby, the roots of $\tau$ are the square roots of those of $\upsilon$. 
This implies that $\tau$ has two real roots if $\upsilon$ has a positive real root. We further observe 
that if $\beta_0$ is a positive real root of $\upsilon(\beta)$, then $\sqrt{\beta_0}$ and $-\sqrt{\beta_0}$ are the real roots of $\tau$. 

Let us now consider the scenario that leads to greatest upper bound on the number of connected
components of $\ca H$.
This scenario corresponds to the situation where $\phi(\beta)$ has
$2n-1$ real roots; let these roots be $\{0, \beta_1, -\beta_1, \dots, \beta_{n-1}, -\beta_{n-1}\},$ 
where $\beta_j \in \bb R_+$ for each $j$. These roots can be mapped to feedback gains via  relation~\eqref{eq:myeq1_c}; in fact, $\beta_j$ and $-\beta_j$ are mapped to the same $k$. 
Thus adding the $k$ that corresponds to the $0$ root via relation~\eqref{eq:myeq1_c}, we will have at most $n$ values for $k$. These values on the other hand, divide the real line into $n+1$ intervals. But only one of the two unbounded intervals can be stabilizing by Observation~\ref{obs:condition_unbounded}; the upper bound of $n$ 
on the number of connected components of $\ca H$ now follows.
\end{proof}
\begin{remark}
  \label{remark:continuous_tangent}
  Note that having $n$ connected components for $\ca H$ implies that we can have a situation
  where $\{k_1, \dots, k_n\}$ are marginally stabilizing gains and the open intervals $(k_j, k_{j+1})$ are stabilizing. 
  %This appear  appears to be a rather extreme scenario. But 
  Figure~\ref{fig:continuous_tangent} demonstrates this phenomena for two adjacent intervals: there is a gain $k_0$ such that the closed-loop system is marginally stable but $(k', k_0)$ and $(k_0, k'')$ are both stabilizing for some $k', k'' \in \bb R$.\footnote{Note that Theorem $1$ in~\cite{feng2018on} uses the same line of reasoning as the proof of 
  Lemma~\ref{lemma:2n_bound_continuous}
  to arrive at the improved bound of  $\ceil{\frac{n}{2}}$ for the number of connected components in $\ca H$,
  assuming that the intervals constructed above are stabilizing/non-stabilizing interlacing intervals. However, as Figure~\ref{fig:continuous_tangent} depicts, this assumption is not valid in general. The $\ceil{\frac{n}{2}}$ bound is still valid, however, and can be obtained by utilizing the structure of the polynomial $\phi(\beta)$ and the relation between the feedback gain $k$ and parameter $\beta$.} The system parameters\footnote{These parameters are actually chosen carefully according to the analysis of Lemma~\ref{lemma:n_bound_continuous}.} are
{
\begin{align*}
  A = \begin{pmatrix}
    -0.825 & -1.21 & -\frac{625919}{4800000} \\
    1 & 0 & 0 \\
    0 & 1 & 0
  \end{pmatrix}, \; b = \begin{pmatrix}
    1 \\ 0 \\ 0
  \end{pmatrix}, \; 
  c= \begin{pmatrix}
    1 \\ 7.5 \\ 12.5
  \end{pmatrix}.
\end{align*}
}
  \end{remark}
  \begin{figure}[ht]
    \centering
    \input{continuous_tangent.tex}
    \caption{An example where the feedback gain $k_0$ leads to a marginally stable closed-loop system yet both intervals $(k', k_0)$ and $(k_0, k'')$ are stabilizing for some $k', k'' \in \bb R$.}
    \label{fig:continuous_tangent}
  \end{figure}
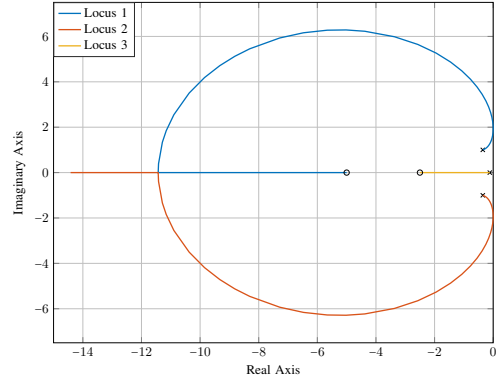
We now proceed to show that the bound on the number of connected components of $\ca H$ can be tightened to $\ceil{\frac{n}{2}}$. The proof of this tighter bound follows a distinct line of reasoning, as necessitated by examples such as that shown in Figure~\ref{fig:continuous_tangent}.
The result contains several technical details. Let us outline the idea behind the proof first: we denote by
$\{0, \beta_1, -\beta_1, \dots, \beta_n, -\beta_n\}$ as the real roots of $\phi(\beta)$ and $\{k_1, \dots, k_{n}\}$ in an increasing order as the feedback gains acquired via relation~\eqref{eq:myeq1_c} in Lemma~\ref{lemma:2n_bound_continuous}:
\begin{enumerate}
  \item We will first show that when two adjacent intervals $(k_{j-1}, k_j)$ and $(k_j, k_{j+1})$ are stabilizing, then $p(\lambda, k_j) = \lambda^n + (a_{n-1}-k_j c_{n-1}) \lambda^{n-1} + \dots + (a_0 - k_j c_0)$ would have the non-stable mode $\lambda_0 = i\beta$, i.e., the mode with zero real part, as a simple root of $p(\lambda, k_j) \in \bb R[\lambda]$.
 By Lemma~\ref{lemma:smooth_dependence}, this would imply that we can find some $C^{\infty}$ function 
  $\eta: {\cal I} \to \bb C$ (with ${\cal I} \subset \bb R$ and $k_i \in {\cal I}$)
 such that $\eta(t)$ tracks the zero of $p(\lambda,k)$ locally with $\eta(k_j) = r$. 
 Indeed, what we really need is that the curve of the root $i\beta$ is differentiable at $i\beta$.
    \item If two adjacent intervals are both stabilizing, the curve $\eta(t)$ is tangent to the imaginary axis at $t_0$. We  show that this observation leads to having $-\beta_j$ and $\beta_j$ as multiple zeros of $\phi(\beta)$. Mind the subtlety here: $\lambda_0$ is the simple root of the polynomial $p(\lambda, k_j) \in \bb R[\lambda]$ in $\lambda$ whereas $\pm \beta_j$ are multiple zeros of the polynomial $\phi(\beta) \in \bb R[\beta]$ in $\beta$ (recall the expression for $\phi(\beta)$ in the proof of Lemma~\ref{lemma:2n_bound_continuous}.). See Figure \ref{fig:relation} for a demonstration of the relations.
      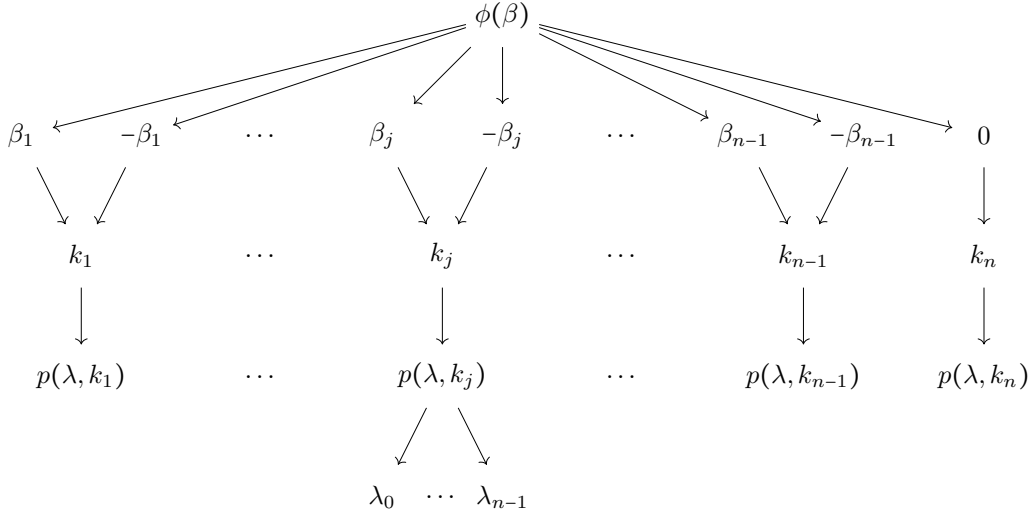
\begin{figure*}
\begin{center}
\begin{tikzpicture}[scale=.8]
\GraphInit[vstyle=Empty]
\tikzset{EdgeStyle/.style = {->}}
\tikzset{VertexStyle/.style = {minimum size = 24 pt}}
\Vertex[Math, L=\phi(\beta{})]{A}
\Vertex[Math, L=\beta_1, x=-8,y=-2]{B}
\Vertex[Math, L=-\beta_1,x=-6,y=-2]{C}
\Vertex[Math, L=\dots, x=-4,y=-2]{D}
\Vertex[Math, L=\beta_j,x=-2,y=-2]{E}
\Vertex[Math, L=-\beta_j,x=0,y=-2]{F}
\Vertex[Math, L=\dots, x=2,y=-2]{G}
\Vertex[Math, L=\beta_{n-1}, x=4,y=-2]{H}
\Vertex[Math, L=-\beta_{n-1}, x=6,y=-2]{I}
\Vertex[Math, L=0, x=8,y=-2]{J}
\Vertex[Math, L=k_1, x=-7,y=-4]{K}
\Vertex[Math, L=\dots, x=-4,y=-4]{L}
\Vertex[Math, L=k_j, x=-1,y=-4]{M}
\Vertex[Math, L=\dots, x=2,y=-4]{N}
\Vertex[Math, L=k_{n-1}, x=5,y=-4]{O}
\Vertex[Math, L=k_n, x=8,y=-4]{P}
\Vertex[Math, L={p(\lambda{}, k_1)}, x=-7,y=-6]{KK}
\Vertex[Math, L=\dots, x=-4,y=-6]{LL}
\Vertex[Math, L={p(\lambda, k_j)}, x=-1,y=-6]{MM}
\Vertex[Math, L=\dots, x=2,y=-6]{NN}
\Vertex[Math, L={p(\lambda,k_{n-1})}, x=5,y=-6]{OO}
\Vertex[Math, L={p(\lambda, k_{n})}, x=8,y=-6]{PP}
\Vertex[Math, L={\lambda_0}, x=-2,y=-8]{MM1}
\Vertex[Math, L=\dots, x=-1,y=-8]{MM2}
\Vertex[Math, L=\lambda_{n-1}, x=0,y=-8]{MM3}
\Edges(A, B)
\Edges (A, C)
% \Edges( A, D)
\Edges[style={->}](A,E) 
\Edges( A, F)
% \Edges( A, G)
\Edges( A, H)
\Edges( A, I)
\Edges( A, J)
\Edges( B, K)
\Edges( C, K)
% \Edges( D, L)
\Edges( E, M)
\Edges( F, M)
% \Edges( G, N)
\Edges( H, O)
\Edges( I, O)
\Edges( J, P)
\Edges( K, KK)
%\Edges( L, LL)
\Edges( M, MM)
% \Edges( N, NN)
\Edges( O, OO)
\Edges( P, PP)
\Edges( MM, MM1)
% \Edges( MM, MM2)
\Edges( MM, MM3)
\end{tikzpicture}
\end{center}
\caption{\label{fig:relation}The relation between roots of $\phi(\beta)$ and $p(\lambda,k)$, as the root locus intersects the imaginary axis for continuous time systems.}
\end{figure*}
      \item Using the multiplicities of $\lambda_j\,'s$ as roots of $\phi(\beta)$ and a careful counting of the stabilizing/non-stabilizing intervals lead to the final result.
\end{enumerate}
We shall developing several propositions before we prove the main result. First we provide an asymptotic expansion of the zeros of $p(\lambda, k)$ with respect to $k$. Note that this is not simply a Taylor expansion around $k$, as a multiple root is not differentiable with respect to $k$.
% {\color{red}
% Variations of these roots due to a small perturbation can be investigated using Puiseux series~\cite{boyd2014solving}.
% It should be noted that the Taylor series is a special case of the Puiseux series for a simple root.%~\cite{avrachenkov2013analytic}.
% }
Recall the definitions of the polynomials $r(\lambda)$ and $s(\lambda)$ (\ref{rs}) used in the proof of Lemma~\ref{lemma:2n_bound_continuous}.
%\begin{align*}
%  r(z) &\coloneqq c_{n-1}z^{n-1} + c_{n-2} z^{n-2} + \dots + c_0,\\
%  s(z) &\coloneqq a_{n-1}z^{n-1} + a_{n-2} z^{n-1} + \dots + a_0.
%\end{align*}
\begin{proposition}
  \label{prop:asymptotic_root}
Suppose for $k_0 \in \bb R$, $\lambda_0$ is a root of $p(\lambda, k_0) \in \bb R[\lambda]$ with multiplicity $m \in \bb N$ and $r(\lambda_0) \neq 0$. Then for $k$ that is sufficiently close to $k_0$,  $p(\lambda, k)$ will have $m$ roots given by,
  \begin{equation*}
    \lambda_j = \lambda_0 + \left((k-k_j) \frac{r(\lambda_0)}{h(\lambda_0)}\right)^{1/m} \omega_j + o(|k-k_j|^{1/m}),
  \end{equation*}
  for $j=1, \dots, m$,
  where $\omega_j$'s are the $m$th roots of unity,\footnote{That is, there are the zeros of $z^m - 1$.}
  $h(\lambda) \in \bb R[\lambda]$ is such that $p(\lambda, k_0) = (\lambda-\lambda_0)^m h(\lambda)$,
and $o(|k-k_j|^{1/m})$ signifies a function $f(k)$ for which $\lim_{k \to k_j} |f(k)|/{|k-k_j|^{1/m}} = 0.$\footnote{Recall root locus rules!}
\end{proposition}  {Before we prove the proposition, let us remark that in the case that $(k-k_j)r(\lambda_0)/h(\lambda_0)$ is negative, there would be $m$ choices for $((k-k_j)r(\lambda_0)/h(\lambda_0))^{1/m}$, namely $(-(k-k_j)r(\lambda_0)/h(\lambda_0))^{1/m} e^{i (\pi + 2 \pi l)/m}$ for $l=\{0, 1, \dots, m-1\}$. Note that these numbers differ multiplicatively from each other by $e^{i2\pi \nu/m},$ with $\nu \in \{0, \dots, m-1\}$; the expression above would not be affected by selecting any of these numbers. Moreover, this statement should be interpreted separately for $k \uparrow k_j$ and $k \downarrow k_j$. The difference amounts to a rotation of $\lambda_j$'s.}
\begin{proof}
  Since $\lambda_0$ is a root of multiplicity $m$, $p(\lambda, k_0) = (\lambda-\lambda_0)^m h(\lambda)$,
  where $h(\lambda) \in \bb R[\lambda]$ with $h(\lambda_0) \neq 0$. 
  Hence, for $k \in \bb R$, we can write $p(\lambda, k)$ as
  \begin{align*}
    p(\lambda, k) &= \lambda^n + (a_{n-1} - k c_{n-1}) \lambda^{n-1} + \dots + a_0 - k c_0 \\
              &= \lambda^n + (a_{n-1}\lambda^{n-1} - k_0 c_{n-1}) \lambda^{n-1} + \dots + a_0 - k_0 c_0 \\
              &\quad - (k-k_0)(c_{n-1} \lambda^{n-1} + \dots + c_0)\\
            &= p(\lambda, k_0) - (k-k_0) r(\lambda) \\
            &= (\lambda-\lambda_0)^m h(\lambda) - (k-k_0)r(\lambda).
  \end{align*}
  Now suppose that $(\lambda-\lambda_0)^m h(\lambda) - (k-k_j) r(\lambda) = 0$. 
  Putting $\hat{\lambda} = \lambda-\lambda_0$, $\hat{h}(\hat{\lambda}) = h(\hat{\lambda}+\lambda_0)/h(\lambda_0)$, $\hat{r}(\hat{\lambda}) = r(\hat{\lambda}+\lambda_0)/h(\lambda_0)$ and $t = k-k_j$, it suffices to show that the zeros of $\hat{\lambda}^m \hat{h}(\hat{\lambda}) - t \, \hat{r}(\hat{\lambda})$ are exactly,
  \begin{align*}
    \hat{\lambda}_j = \left( \frac{r(\lambda_0)}{h(\lambda_0)}\right)^{1/m} \omega_j + o(|t|^{1/m}), \text{ for } j=1, \dots, m,
  \end{align*}
  as $t \to 0$. Note that $\hat{h}(0) = 1$ and $\hat{h}(\hat{\lambda}) \in \bb R[\hat{\lambda}]$. 
  %{In the case that $r(\lambda_0)/h(\lambda_0)$ is negative, there would be $m$ choices for $(r(\lambda_0)/h(\lambda_0))^{1/m}$, namely $(-r(\lambda_0)/h(\lambda_0))^{1/m} \omega_j$, since one may choose any number satisfying $z^m = r(\lambda_0)/h(\lambda_0)$ in $\bb C$. Note that these numbers differ multiplicatively from each other by $\omega_j^{l}$ , with $l \in \{0, \dots, m-1\}$; the expression above would not be affected by selecting any of these numbers.}\par
 Let  $t > 0$ and observe that if $z$ is a zero of
  $$\psi_t(\lambda) \coloneqq \lambda^m \hat{h}(t^{1/m}\lambda) - \hat{r}(t^{1/m}\lambda), $$ then $t^{1/m}z$ would be a zero of $\lambda^m h(\lambda) - t \, r(\lambda)$. 
On a compact set $[0, T]$ ($T \in \bb R_+$), $\psi_t(\lambda) \to \psi_0(\lambda) = \lambda^M - r(0)$ uniformly. Let us denote the zeros of $\psi_0(\lambda) = \lambda^m - r(0)$ by
  \begin{align*}
    z_j = \left(r(0)\right)^{1/m} \omega_j, \text{ for }j=1, 2, \dots, m,
  \end{align*}
  where $\omega_j$'s are the $m$th roots of unity. Choose a sufficiently small $\varepsilon > 0$ such that the disks $B_{z_j}(\varepsilon)$ are disjoint. Since $\partial B_{z_j}(\varepsilon)$ is compact and $\psi_0(\lambda)$ does not vanish on $\partial B_{z_j}(\varepsilon)$, there is some $l > 0$ such that $|\psi_0(\lambda)| > l$. Since $\psi_t(\lambda) \to \psi_0(\lambda)$ uniformly on any compact subset of $\bb C$, there is some $t^* > 0$, such that $|\psi_t(\lambda) - \psi_0(\lambda)| < l$ for $t \in (-t^*, t^*)$ and $\lambda \in \bar{B}_{z_j}(\varepsilon)$. By Rouch\'{e}'s Theorem~\cite{rudin2006real}, there is exactly one zero for $\psi_t(\lambda)$ in each $B_{z_j}(\varepsilon)$ for $t \in (-t^*, t^*)$. As such, the zeros of $\psi_t(\lambda)$ are given by,
  \begin{align*}
   r(0)^{1/m} \omega_j + o(1).
  \end{align*}
  It then follows that,
  \begin{align*}
    \lambda_j' = \left( t\frac{r(\lambda_0)}{h(\lambda_0)}\right)^{1/m} \omega_j + o(|t|^{1/m}), \text{ for } j=1, \dots, m;
   \end{align*}
  the case where $t < 0$ can be treated similarly by considering $t \, h(\lambda) = (-t)(-h(\lambda))$.
Hence, 
  \begin{align*}
    \lambda_j = \lambda_0 + \left( (k-k_j) \frac{r(\lambda_0)}{h(\lambda_0)}\right)^{1/m} \omega_j + o(|k-k_j|^{1/m}), 
  \end{align*} for  $j=1, \dots, m$.
\end{proof}
Before we proceed to the proof of the main result of this section,
let us make an observation pertaining to the derivative of $f(\beta) \in \bb R[\beta]$ 
evaluated on the imaginary axis.
\begin{proposition}
  \label{prop:poly_derivative_continuous}
   Consider $f(\lambda) = \alpha_{m} \lambda^m + \dots + \alpha_1 \lambda + \alpha_0 \in \bb R[\lambda]$ 
%    is a polynomial, 
  and let $\varphi(\beta) = \text{\bf Im}\left(f(\lambda)\big \vert_{\lambda=i\beta}\right) \in \bb R[\beta]$.
  Then we have,
    \begin{align*}
      \frac{d \varphi(\beta)}{d\beta} = \text{\bf Re}\left( \frac{d f}{d\lambda}\Big \vert_{\lambda=i\beta}\right).
    \end{align*}
\end{proposition}
\begin{proof}
 When $m$ is odd we have,
  $$\varphi(\beta) = i^{m-1} \alpha_m \beta^m + i^{m-3} \alpha_{m-2} \beta^{m-2} + \dots + i^2 \alpha_{3} \beta^{3} + \alpha_1 \beta.$$
Hence,
\begin{align*}
  \varphi'(\beta) &= m i^{m-1} \alpha_m \beta^{m-1} + (m-2) i^{m-3} \alpha_{m-2} \beta^{m-3} + \cdots \\
  &\quad + 3 i^2 \alpha_3 \beta^2 + \alpha_1.
\end{align*}
But we also note that,
\begin{align*}
  \text{\bf Re} \left( \frac{df}{d\lambda} \Big\vert_{\lambda=i\beta}\right) &= \text{\bf Re}\left( \left(m\alpha_m x^{m-1} + \cdots + \alpha_1\right)|_{x=i\beta}\right) \\
                                                             &= \text{\bf Re} \left( m \alpha_m (i\beta)^{m-1} + \cdots + \alpha_1\right) \\
                                                             &= m i^{m-1} \alpha_m \beta^{m-1} + (m-2)i^{m-3} \alpha_{m-2} \beta^{m-3}\\
                                                             & + \cdots + 3 i^2 \alpha_3 \beta^2 + \alpha_1 \\
  &= \varphi'(\beta)
\end{align*}
The case when $m$ is even follows analogously.
\end{proof}
We are now ready to prove the bound $\ceil{\frac{n}{2}}$ on the number of connected components
of $\ca H$.
\begin{lemma}
  \label{lemma:n_bound_continuous}
Let $\ca H \neq \emptyset$; then it has at most $\ceil{\frac{n}{2}}$ connected components.
\end{lemma}
\begin{proof}
Consider two adjacent intervals $(k_{j-1}, k_j)$ and $(k_j, k_{j+1})$ that are both stabilizing. Note that by assumption $p(\lambda_0, k_j) = 0$ for some pure imaginary $\lambda_0$. First, let us examine whether $\lambda_0$  can have multiplicity $m \ge 2$. By Proposition~\ref{prop:asymptotic_root} for $k$ sufficiently close to $k_0$,  $p(\lambda, k)$ will have $m$ roots given by,
  \begin{align*}
   \lambda_j = \lambda_0 + \left((k-k_j) \frac{r(\lambda_0)}{h(\lambda_0)}\right)^{1/m} \omega_j + o(|k-k_j|^{1/m}) ,
  \end{align*}
  for $j=1, \dots, m$, where $\omega_j$'s are the $m$th roots of unity.
  %, i.e., $z^M = 1$, and $o(|k-k_j|^{1/M})$ denotes that $\lim_{k \to k_j} \frac{o(|k-k_j|)}{|k-k_j|^{1/M}} = 0$.\\
 When $m \ge 3$, as $k \to k_j$, $m$ roots, from $m$ equally spaced directions in the complex plane (see Figure~\ref{fig:equal_space_continuous}), would tend to $\lambda_0$.
  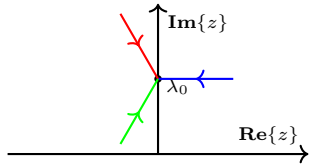
\begin{figure}[ht]
    \centering
    \begin{tikzpicture}
    \begin{scope}[thick,font=\scriptsize, decoration={
    markings,
    mark=at position 0.5 with {\arrow{>}}}]
    % Axes:
    % Are simply drawn using line with the `->` option to make them arrows:
    % The main labels of the axes can be places using `node`s:
    \draw [->] (-2,0) -- (2,0) node [above left]  {${\bf Re} \{z\}$};
    \draw [->] (0,0) -- (0,2) node [below right] {${\bf Im} \{z\}$};

    % Axes labels:
    % Are drawn using small lines and labeled with `node`s. The placement can be set using options
    \iffalse% Single
    % If you only want a single label per axis side:
    \draw (1,-2pt) -- (1,2pt)   node [above] {$1$};
    \draw (-1,-2pt) -- (-1,2pt) node [above] {$-1$};
    \draw (-2pt,1) -- (2pt,1)   node [right] {$i$};
    \draw (-2pt,-1) -- (2pt,-1) node [right] {$-i$};
    \else% Multiple
    % If you want labels at every unit step:
    % \foreach \n in {-4,...,-1,1,2,...,4}{%
    % \draw (\n,-3pt) -- (\n,3pt)   node [above] {$\n$};
    % \draw (-3pt,\n) -- (3pt,\n)   node [right] {$\n i$};
    % }
    \fi
    \node at (-0.1, .9) [circle, right] {$\lambda_0$};
    \draw[fill=black] (0, 1) circle[radius=1pt];
    \draw[postaction={decorate}, blue] (1, 1) -- (0, 1);
    \draw[postaction={decorate}, red] (-0.5, 1.866) -- (0,1);
    \draw[postaction={decorate},green] (-0.5, 0.134) -- (0,1);
    \end{scope}
    % The circle is drawn with `(x,y) circle (radius)`
    % You can draw the outer border and fill the inner area differently.
    % Here I use gray, semitransparent filling to not cover the axes below the circle
    %\path [draw=none,fill=gray,semitransparent] (0, 0) circle (1);
\end{tikzpicture}
    \caption{Roots coming from $m=3$ equally spaced directions in the complex plane tend to $\lambda_0$}\label{fig:equal_space_continuous}
  \end{figure}
  In this case, at least one direction would be in the right half plane, contradicting the assumption
 that both $(k_{j-1}, k_j)$ and $(k_j, k_{j+1})$ are stabilizing. If $m=2$, the roots would traverse a {horizontal line} passing through $\lambda_0$,\footnote{Note that in this case either $k\uparrow k_j$ or $k \downarrow k_j$ would yield the quantity $(k-k_j) r(\lambda_0) /h(\lambda_0)$ in Proposition~\ref{prop:asymptotic_root} positive, and the resulting expression would be a scalar multiple of the primitive $2^{\text{nd}}$ roots of unity, i.e., $\pm \alpha $ for some $\alpha \in \bb R$.}  {with on of the intervals as nonstabilizing}. This is again a contradiction. 
  Hence $\lambda_0$ must be simple. This on the other hand would imply the differentiability of the root with respect to $k$. By the asymptotic formula above (with $m=1$), the derivative of $\eta'(k_j) = {r(\lambda_0)}/{h(\lambda_0)}$; note that in this scenario $h(\lambda_0) = p'(\lambda_0, k_j)$. But the condition that both intervals are stabilizing implies that $\eta'(\lambda_0)$ is pure imaginary (the tangent of the curve should be the imaginary axis, see Figure~\ref{fig:continuous_tangent}); that is
  \begin{align*}
    \frac{d \eta}{dk}\Big \vert_{k_j} = i \gamma,
  \end{align*}
  for some $\gamma \in \bb R$. This implies that if $\lambda_0 = i\beta$ ($\beta \in \bb R$)
  is a zero of $\phi(\beta)$; {since $\eta'(\lambda_0)$ is pure imaginary,} we would also have,
  \begin{align}
    \label{eq:cont_alg_derivative}
    \text{\bf Re} \left( \left(\lambda^n + s(\lambda)-k_j r(\lambda) \right)' r(\bar{\lambda})|_{\lambda = i \beta}\right) = 0,
  \end{align}
  where $s(\lambda) = a_{n-1}\lambda^{n-1}+ \dots + a_0$. 
  Putting $r(\lambda)=r_o(\lambda) + r_e(\lambda)$, where $r_o(\lambda)$ consists of
  the terms with odd degrees and $r_e(\lambda)$ consists terms of even degrees, we observe 
  that $r(\bar{\lambda})|_{\lambda=i \beta} = r(-i \beta)= r_e(i\beta) - r_o(i\beta) = \left( r_e(\lambda)-r_o(\lambda)\right)|_{z = i\beta}$. Now by Proposition~\ref{prop:poly_derivative_continuous}, we have
\begin{align*}
  \phi'(\beta) &= \text{\bf Re} \left[\left( \left(\lambda^n + s(\lambda)\left(r_e(\lambda)-r_o(\lambda)\right)\right)\right)'|_{\lambda=i\beta} \right] \\
               &= \text{\bf Re} \left\{ \left[\left(\lambda^n + s(\lambda)\right)' \left(r_e(\lambda)-r_o(\lambda)\right) \right. \right. \\
               & \quad + \left. \left. \left(\lambda^n + s(\lambda)\right)\left(r_e'(\lambda) - r_o'(\lambda)\right) \right]|_{\lambda=i\beta}\right\} .
\end{align*}
Noting that $\left(\lambda^n + s(\lambda) - k\, r(\lambda)\right)|_{\lambda=\lambda_0} = 0$, we have
\begin{align*}
  \phi'(\beta) &= \text{\bf Re} \left\{ \left[ \left(\lambda^n + s(\lambda)\right)' \left(r_e(\lambda)-r_o(\lambda)\right) \right. \right. \\
&\quad +\left. \left.  \left(-k\,r(\lambda)\right)\left(r_e'(\lambda)-r_o'(\lambda)\right) \right]|_{\lambda=i\beta}\right\}.
\end{align*}
Now $r_e'(\lambda)$ consists of terms with odd degrees and $r_o'(\lambda)$ of those with even degrees;
this implies that,
\begin{align*}
  \left(r_e'(\lambda)-r_o'(\lambda)\right)|_{\lambda=i\beta} &=  \left(-r_e'(\lambda) - r_o(\lambda)\right)|_{\lambda=-i\beta} \\&= -r'(\bar{\lambda})|_{\lambda=i\beta},
\end{align*}
minding that $r'(\bar{\lambda})$ refers to the derivative of $r'(\lambda)$ evaluated at $\bar{\lambda}$.
Furthermore, 
\begin{align*}
  \text{\bf Re}\left[\left( -k r'(\lambda) r(\bar{\lambda})\right)|_{\lambda=i\beta} \right] = \text{\bf Re} \left[ \left(-kr'(\bar{\lambda}) r(\lambda)\right)|_{\lambda=i\beta}\right].
\end{align*}
Combining all these observations, we conclude that $\phi'(\beta) = 0$. 

Now let $\{0, u_1,-u_1, \dots, u_\mu,-u_{\mu}, v_1, -v_1, \dots, v_\nu, -v_\nu\}$ denote the roots of the polynomials $\phi(\beta)$, where $u_i$'s and $v_j$'s are nonnegative real numbers, and $v_1, -v_1, \dots, v_\nu, -v_\nu$ are roots with multiplicity greater than $1$. We must then have $2\mu + 2(2\nu) \le 2n-2$. These roots will be mapped to the gains $k$ via the relation,
\begin{align*}
  k = \frac{q(\lambda)}{\lambda^n + r(\lambda)}\Big \vert_{\lambda=i\beta}.
\end{align*}
Now let $\Lambda = \{k_1, \dots, k_{\mu'}\}$ be the set of {distinct real} gains corresponding to $\{u_1, -u_1, \dots, u_{\mu}, -u_{\mu}\}$ and $\Pi = \{k_1, \dots, k_{\nu'}\}$ be the set of {distinct real} gains corresponding to $\{v_1, -v_1, \dots, v_{\nu}, -v_{\nu}\}$. Note  that $\mu' \le \mu$ and $\nu' \le \nu$ since it is possible that multiple roots are mapped to the same $k$; append $k$ corresponding to $0$ via~\eqref{eq:myeq1_c}. We must then have $\mu' + 2 \nu' \le n$. Now if $k_j \in \Lambda$, then one of the intervals $(k_{j-1}, k_j)$ and $(k_j, k_{j+1})$ is not stabilizing. Let $\Upsilon$ be the collection of intervals that are not stabilizing. It follows then that if $k \in \Lambda$, $k$ must be the end point of an interval in $\Upsilon$. Note that only one of the unbounded intervals could be stabilizing; thereby,
\begin{align*}
  \mu' \le 1 + 2(|\Upsilon| - 1) = 2|\Upsilon| - 1.
\end{align*}
{Now, Lemma~\ref{lemma:2n_bound_continuous} implies that the set of feedback (real) gains $k$ is divided into $\nu' + \mu' +1$ intervals. As such, $\nu' + \mu' +1 - |\Upsilon|$ is the number of stabilizing intervals, obtained by subtracting the number of  non-stabilizing intervals from the total number of intervals. Hence,}
%Lastly, $\nu' + \mu' +1 - |\Upsilon|$ is the number of stabilizing intervals. Hence,
\begin{align*}
  \mu' + \nu' +1 - |\Upsilon| &\le \mu' + \nu' + 1 - \frac{\mu' + 1}{2} \\ &= \nu' + \frac{\mu' + 1}{2} \le \frac{2 \nu' + \mu' + 1}{2} \le \frac{n+1}{2} = \ceil{\frac{n}{2}}.
\end{align*}
\end{proof}
\begin{remark}
  We note that $\ceil{\frac{n}{2}}$ is a tight upper bound. Indeed, Figure~\ref{fig:continuous_tangent} already indicates that there are two disjoint stabilizing intervals for a SISO system with $n=3$. Figure~\ref{fig:connected_components_cont} provides a more transparent view of this; the system parameters are,
  \begin{align*}
    A = \begin{pmatrix}
      0 & 1 & 0 \\
      0 & 0 & 1 \\
      -0.133 & -1.125 & -0.625
    \end{pmatrix}, \; b = \begin{pmatrix}
      0 \\
      0 \\
      1
    \end{pmatrix}, \; c=\begin{pmatrix}
      12.5 \\
      7.5 \\
      1
    \end{pmatrix}.
  \end{align*}
  \begin{figure}[ht]
   \centering
   \input{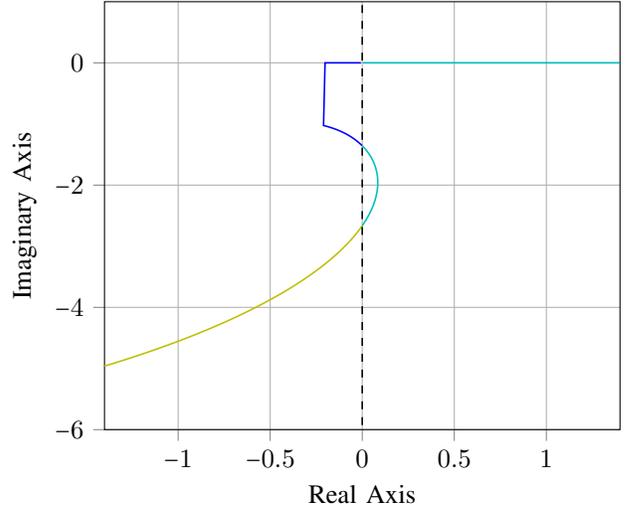}
   \caption{\label{fig:connected_components_cont}The figure depicts how the root with the largest real part varies with respect to the feedback gain $k$. The blue and yellow segments correspond to two stabilizing intervals.}
  \end{figure}
  \end{remark}
%%%%
{We mention that although 
%the assumption that
% in~\cite{feng2019}, i.e., 
stabilizing and non-stabilizing intervals do not necessary interlace for all controllable and observable triplets $(A, b, c^\top)$, this property is indeed generic.\footnote{A property is generic when it holds except on an algebraic (Zariski closed) set.} To this end, we first observe that the set $\ca{CO}$ of all controllable and observable triplets $(A, b, c^{\top})$ is open in $\bb M_{n \times n}(\bb R)\times \bb R^n \times \bb R^n$.
  \begin{proposition}
    The set $\ca{CO}$ is (Zariski) open.
  \end{proposition}
  \begin{proof}
    Putting $\bf C$ to denote the controllability matrix and $\bf O$ to denote the observability matrix. We observe $\ca{\left(CO\right)}^{c} = \ca C^c \cup \ca O^c \cup \ca C^c \ca O^c$, where $\ca C^c$ denotes the collection of non-controllable but observable systems, $\ca O^c$ denotes the collection of non-observable but controllable systems, and $\ca C^c \ca O^c$ denotes the collection of non-controllable and non-observable systems.\footnote{For a set $\ca A$, $\ca A^c$ denotes its complement.} We note that $\ca C^c$ is the set where all the $n \times n$ minors of $\bf C$ vanish. Similarly for $\ca O^c$ and $\ca C^c\ca O^c$. Hence, $\ca{(CO)}^c$ is an algebraic set. Consequently, $\ca{CO}$ is open.
  \end{proof}
  We now observe the interlacing property is generic.
  \begin{lemma}
    \label{lemma:interlacing}
    For a controllable and observable system $(A, b, c^\top)$, the property that the stabilizing and non-stabilizing intervals in $\ca H$ interlace is generic.
  \end{lemma}
  \begin{proof}
 
    Denote $\ca U \subseteq \ca{CO}$ the subset such that the corresponding Hurwitz stabilizing set $\ca H$ has the interlacing property.
    As we have shown, for any $(A, b, c^\top) \in \ca{CO}$, if we have two adjacent stabilizing (or non-stabilizing) intervals, then the corresponding polynomial $\phi(\beta)$ must have a root with multiplicity greater than one. This means the discriminant of $\phi(\beta)$ must vanish. This in turn is a polynomial in the entries of $A, b, c^\top$. That is,
\begin{align*}
  \ca U^c = \{(A, b, c^\top) \in \ca U^c: \Delta(\phi(\beta)) = 0\}.
\end{align*}
This suggests that $\ca U$ is Zariski open and nonempty (see for example, Figure~\ref{fig:connected_components_cont}).\footnote{Note that $\bb M_n(\bb R) \times \bb R^n \times \bb R^n$ is identified by the affine space $\bb A^{n^2 + 2n}[\bb R]$ and the subset $\ca{CO}$ is equipped with the Zariski subspace topology.}
  \end{proof}
}
%%%%%%  
\subsubsection*{\bf An algorithm for characterizing the connected components of $\ca H$}
Our analysis for deriving the bound $\ceil{\frac{n}{2}}$ for the number of connected
components of $\ca H$ has direct algorithmic implications.
%clearly yields an algorithm on how to identify these stabilizing intervals. 
 We summarize the corresponding algorithm as follows.
 \begin{algorithm}[H]
\floatname{algorithm}{Algorithm $1$:}
\renewcommand{\thealgorithm}{}
\caption{\bf Identifying stabilizing intervals of $\ca H$}
\label{cont_alg}
\begin{algorithmic}[1]
\STATE Find the real roots $\{\lambda_1, \dots, \lambda_l\}$ of the real polynomial~\eqref{eq:alg_poly_cont}. Appending $\{0\}$ to this list if necessary, we get $L = \{0, \lambda_1, \dots, \lambda_l\}$. Map $L$ to $\{k_1, \dots, k_{l'}\}$ (order this list in an increasing manner) by~\eqref{eq:myeq1_c}.
\STATE Identify whether $(-\infty, k_1)$ and $(k_{l'}, +\infty)$ are stabilizing (Observation~\ref{obs:condition_unbounded}).
\STATE If $(-\infty, k_1)$ is stabilizing, check the multiplicity of $\lambda_{1'}$ that maps to $k_1$. If $\lambda_{1'}$ is simple, then $(k_1, k_2)$ is not stabilizing. If $\lambda_{1'}$ is not simple, check whether~\eqref{eq:cont_alg_derivative} is satisfied; if not, i.e., the corresponding derivative is not pure imaginary, then $(k_1, k_2)$ is not stabilizing. If this derivative is pure imaginary, then $(k_1, k_2)$ is stabilizing. Continue the process.
\end{algorithmic}
\end{algorithm} 
The main computational cost of Algorithm $1$ is finding the roots of a real polynomial. The specifics are beyond the scope of this paper; see~\cite{pan2016nearly, kobel2016computing} and references therein for the recent algorithmic developments in this direction. \par
{Let us demonstrate the progression of Algorithm $1$ for the example in Remark~\ref{remark:continuous_tangent}.
%\begin{itemize}
%\item 
The characteristic polynomial of the closed-loop system for this example is then
  \begin{align*}
   p(\lambda, k) = \lambda^3 +0.825 \lambda^2 + 1.21 \lambda + 0.3401666667 \\
    \quad +k (\lambda^2 + 7.5\lambda + 12.5).
  \end{align*}
%\item 
Furthermore, 
  \begin{align*}
    \phi(\beta) = \lambda^5 - 7.5225 \lambda^3 + 1.367899792 \lambda,
  \end{align*}
with nonnegative roots $0, \frac{\sqrt{6018}}{40},$ where $\frac{\sqrt{6018}}{40}$ has multiplicity $2$.
%\item
The root $\beta_1 =0$ is mapped to $k_1=-\frac{625919}{6 \times 10^7}$, and $\beta_2 =\frac{\sqrt{6018}}{40}$ to
$k_2 = \frac{2041}{6000}$.
%\item 
We start from the unbounded interval $(k_2, \infty)$: pick $k$ and test the stability of the closed loop system; n this case, we conclude that $(k_2, \infty)$ is stabilizing.
  %\item 
  Since $k_2$ is acquired from a multiple root of $\phi(\beta)$, we examine the derivative $p'(\lambda, k_2)$;
  since it is pure imaginary, we conclude that the interval $(k_1, k_2)$ is stabilizing.
    %\item
     As such $(-\infty, k_1)$ is not stabilizing by Observation~\ref{obs:condition_unbounded}.
%\end{itemize}
}
\section{Properties of Schur stabilizing feedback gains}
\label{sec:schur}
In this section, we study the properties of the set of static feedback gains~\eqref{eq:discrete_S} for discrete-time linear systems~\eqref{eq:discrete_system}.
%We now consider the discrete-time LTI SISO system
%\begin{align*}
%  &x(k+1) = Ax(k) + b u(k), \\
%  & y(k) = c^\top x(k),
%\end{align*}
%where $A \in M_n(\bb R)$ and $b,c \in \ca M(n \times p; \bb R)$. We are interested in investigating the properties of the set of all static feedback controllers, i.e.,
%\begin{align*}
%  \ca S = \{ k \in \bb R: \rho(A-k bc^\top) < 1\}.
%\end{align*}
Here is our first observation.
%We first observe that for state feedback SISO systems, there is a homeomorphism 
%between the sets $\ca S_{\bf x}$ and $\ca H_{\bf x}$.
\begin{lemma}
  \label{lemma:homeo_d_c}
  There is a homeomorphism $h: \ca H_{\bf x} \to \ca S_{\bf x}$.
\end{lemma}
\begin{proof}
Without loss of generality, as we have done throughout this paper, we assume that the pair $(A,b)$ is in {the controllable canonical form}.
Recall that the bilinear transform,
$$g:\lambda \mapsto \frac{\lambda+1}{\lambda-1}$$ is a diffeomorphism between the unit disk $\bb D$ and the open left-half plane $\bb H_{-}$ in $\bb C$. Clearly $G \coloneqq (g, \dots, g): \bb D^n \to \bb H_{-}^n$ defines a diffeomorphism between $\bb D^n \to \bb H_{-}^n$. Passing to the quotient space (modulo the action of the symmetric group), we have a diffeomorphism $\tilde{G}: \bb D^n/S_n \to \bb H_{-}^n/S_n$ given by $\tilde{G} \circ \pi = G$, where $\pi$ is the canonical projection.  Let $\zeta$ be the bijection between $\bb R^n$ and the set of monic $n$th degree polynomials, i.e., if $\alpha = (\alpha_0, \dots, \alpha_{n-1}) \in \bb R^n$, then $\zeta(\alpha) = \lambda^n + \alpha_{n-1} \lambda^{n-1} + \dots + \alpha_0$. Denote the sets,
\begin{align*}
  \ca E &= \{ \alpha \in \bb R^n: \zeta(\alpha) \text{ has all roots in $\bb D$}\},\\
          \ca F &= \{ \alpha \in \bb R^n: \zeta(\alpha) \text{ has all roots in $\bb H_{-}$}\}.
\end{align*}
By Corollary~\ref{cor:real_zeros}, we have following commuting diagram:
\begin{figure}[H]
  \centering
        \begin{tikzcd}
          \ca E \arrow[d, "\hat{\sigma}"] \arrow[r, dashrightarrow, "\hat{\sigma}^{-1} \circ \tilde{G} \circ \sigma"]  & \ca F \arrow[d, "\hat{\sigma}"] \\
\{\bb D^n\}_{*}/S_n \arrow[r, "\tilde{G}"] & \{\bb H^n_{-}\}_{*}/S_n,
        \end{tikzcd}
 % \caption{\label{fig:label} }
\end{figure}
\noindent where $\{\bb D^n\}_{*}/S_n \subseteq \bb C^n/S_n$ denotes the $n$-dimensional
vector invariant under conjugation with entries inside the unit disk of $\bb C$; 
similarly for $\{\bb H^n_{-}\}_{*}/S_n$ (recall the notation in \S~\ref{sec:notations}). 

Now, the map
\begin{align*}
  \hat{\sigma}^{-1} \circ \tilde{G} \circ \hat{\sigma}: \ca E \to \ca F,
\end{align*}
defines a homeomorphism between $\ca E$ and $\ca F$. 
However, note that $\ca S_{\bf x} = a - \ca E$ and $\ca H_{\bf x} = a - \ca F,$ 
where $a$ is the last row of $A$. This completes the proof since a translation in $\bb R^n$  is 
a diffeomorphism.
\end{proof}
\begin{remark}
    The above result immediately implies that $\ca S_{\bf x}$ is open, connected and contractible since $\ca H_{\bf x}$ is. In our subsequent discussion, we will also outline more direct proofs for the above facts since they provide additional insights into the structure of the set of stabilizing feedback gains for discrete-time linear systems.
\end{remark}
In view of Lemma~\ref{lemma:homeo_d_c}, one might be inclined to construct a similar homeomorphism between 
the sets $\ca S$ and $\ca H$. However, as it turns out, the technique adopted in Lemma~\ref{lemma:homeo_d_c} can not be generalized for this purpose. 
For example, when $k_0 \in \ca S$, it does follow that $k_0 c \in \ca S_{\bf x}$. However,
under the homeomorphism constructed in Lemma~\ref{lemma:homeo_d_c}, the image of $k_0 c$ is not necessarily a scalar multiple of $c$. For a concrete example, one may consider the triplet $(A, b, c^\top)$
%where the last row of $A^{\flat}$ is zero; we now notice that $(A^{\flat}-kb^{\flat}c^T - I)^{-1}(A^{\flat}-kb^{\flat}c^T + I)$ can not be represented by $A^{\flat}-k' b^{\flat} c^T$ for some $k' \in \bb R$. \footnote{Note that we are not claiming that such homeomorphism does not exist. The non-existence of such a homeomorphism requires a deeper understanding of these two topological subspaces. To the best of our knowledge such a homeomorphism has not been reported in the literature.} 
given by
\begin{align*}
A = \begin{pmatrix}
  0 & 1 & 0 & 0 \\
  0 & 0 & 1 & 0 \\
  0 & 0 & 0 & 1 \\
  0 & 0 & 0 & 0
\end{pmatrix}, \; b = \begin{pmatrix}
  0 \\
  0 \\
  0 \\
  1
\end{pmatrix}, \; c = \begin{pmatrix}
                     1 \\
                     2 \\
                     3 \\
                     4
                   \end{pmatrix}.
\end{align*}
We note $0 \in \ca S$ and $\ca Z_{p(\lambda, A-0 b c^\top)} = \{0, \dots, 0\}$. Under the bilinear transform, the zeros will be mapped to $\{-1, \dots, -1\}$, which corresponds to characteristic polynomial $p(\lambda)=\lambda^4 - 4 \lambda^3 + 6 \lambda^2 - 4\lambda +1$. However, no $k \in \ca H$ can yield a closed-loop system $A-k b c^\top$ with this characteristic polynomial.\footnote{Note that we are not claiming that such homeomorphism does not exist. The non-existence of such a homeomorphism requires a deeper understanding of these two topological subspaces. To the best of our knowledge such a homeomorphism has not been reported in the literature.} 

We now demonstrate that:
\begin{enumerate}
  \item $\ca S$ and $\ca S_{\bf x}$ are both open in the Euclidean topology.
    \item $\ca S$ and $\ca S_{\bf x}$ are both bounded.
      \item $\ca S$ and $\ca S_{\bf x}$ are convex if the system has two states.
\item $\ca S_{\bf x}$ is connected and $\ca S$ has at most $\ceil{\frac{n}{2}}$ connected components.
\end{enumerate}
Most of the proofs for the discrete time case have a similar flavor as their continuous counterparts. 
However, the proof for the upper bound on the number of connected components of $\ca S$ 
has a few distinct steps. 
\begin{lemma}
The set $\ca S$ is open in $\bb R$ and $\ca S_{\bf x}$ is open in $\bb R^n$.
\end{lemma}
\begin{proof}
  The proof proceeds similar to the proof of Lemma~\ref{lemma:open_cont}. We only need to observe
  that the composition map, $$\upsilon \coloneqq \max \circ \; |\cdot|: \bb C_{*}^n \to [0, \infty),$$ is continuous, even when adopted on the quotient space. That is, there is a unique continuous map $\tilde{\upsilon}: \bb C_{*}^n/S_n \to [0, \infty)$ such that $\upsilon = \tilde{\upsilon} \circ \pi$. 
  Hence the map,
\begin{align*}
  k \mapsto {\cal Z}_{p(\lambda, A-kbc^\top)} \mapsto \max \left(|{\cal Z}_{p{(\lambda, A-kbc^\top)}}|\right)
\end{align*}
is continuous.
The interval $[0, 1)$ is open in the subspace topology of $[0, \infty)$ and $\ca S$ is the preimage of $[0, 1)$ under 
the above map; thereby, $\ca S$ is open. 

In order to show that $\ca S_{\bf x}$ is open in $\mathbb R^n$, 
we only need to observe that the map $F: \bb R^n \to [0, \infty)$, given by
\begin{align*}
  k \mapsto {\cal Z}_{p(\lambda, A-bk^\top)}\mapsto \max \left(|{\cal Z}_{p{(\lambda,A-bk^\top)}}|\right),
\end{align*}
is continuous and $\ca S_{\bf x} = F^{-1}([0, 1))$.
\end{proof}
We now note that contrary to the continuous time case, 
the sets $\ca S$ and $\ca S_x$ are bounded.
\begin{proposition}
  \label{prop:discrete_bounded}
 The set $\ca S_{\bf x}$ is bounded in $\bb R^n$.
\end{proposition}
\begin{proof}
  It suffices to assume that $(A, b)$ is in the {controllable canonical form}. For any $k \in \bb R^n$, the characteristic polynomial
  of the corresponding closed-loop system assumes the form,\footnote{The entries of $k$ are consistent with the way the are indexed in the characteristic polynomial.}
\begin{align*}
  p(\lambda, k) = \lambda^{n} + (a_{n-1} - k_{n-1})\lambda^{n-1}+ \dots + (a_0 - k_0).
\end{align*}
Let $\lambda_1, \dots, \lambda_n$ denote the zeros of $p(\lambda,k)$.
By Vieta's formula, the coefficients of $p(\lambda, k)$ are elementary symmetric functions of its roots $\lambda_1, \dots, \lambda_n$. 
For every $k \in \ca S_{\bf x}$, $|\lambda_j| < 1$ for all $j$; as such, the coefficients of 
$p(\lambda, k)$, and by extension, $k_0, \ldots, k_{n-1}$, are bounded.
\end{proof}
\begin{corollary}
  \label{cor:bounded_discrete}
  For the controllable and observable output feedback system $(A, b, c^\top)$, the set $\ca S$ is bounded.
\end{corollary}
\begin{proof}
  We observe that
%\begin{align*}
  $\ca S = \{k \in \bb R: kc \,\cap \,\ca S_{\bf x} \neq \emptyset\}.$
 % \end{align*}
   \end{proof}
In regards to convexity properties, it is known that the set $\ca S_{\bf x}$ is not convex in general~\cite{bialas1985convex, fell1980zeros}. 
%We now remark that when the system only has two states $\ca S_{\bf x}$ and $\ca S$ are convex.
\begin{observation}
  The set $\ca S_{\bf x}$ is convex when $n=2$.
\end{observation}
\begin{proof}  Without loss of generality, we assume that the pair $(A, b)$ is in {the controllable canonical form}. 
Suppose that $k = (k_0, k_1)^\top$ and $e = (k'_0, k'_1)^\top$ are two stabilizing controllers. The characteristic polynomials of the corresponding closed-loop systems are then,
  \begin{align*}
    p_k(\lambda) = \lambda^2 + k_1 \lambda + k_0, \; \mbox{and} \; p_{k'}(\lambda) = \lambda^2 + k'_1 \lambda + k'_0.
  \end{align*}
  Note that by Vieta's formula, $k_0 < 1$ and $k'_0 < 1$ since the zeros are inside the unit disk.
  For ${\hat k} = (1-\delta) k + \delta k'$ with $\delta \in [0, 1]$, consider the corresponding characteristic equation
   $p_{\hat k}(\lambda) = (1-\delta)p_k(\lambda) + \delta p_{k'}(\lambda)$. 
  %   is the given by $p_{\hat k}(\lambda)=\lambda^2 + ((1-\delta)k_1 + \delta k'_1) x + ((1-\delta)k_0 + \delta k'_0)$. 
  %
If $p_{\hat k}$ has two conjugate zeros $z_1, \bar{z}_1$, then $|z_1| = |\bar{z}_1| < 1$ by Vieta's formula since $|z_1|^2 = (1-\delta) k_0 + \delta k'_0 < 1$. On the other hand, if $p_{\hat k}$ has two real zeros, suppose that one of them is $1$ or $-1$ (note that by Vieta's formula, the product of two zeros is strictly less than $1$; hence the other zero is inside the open unit disk), i.e., $p_{\hat k}(1) = 0$ or $p_{\hat k}(-1) = 0$. Note that $p_k(1), p_k(-1), p_{\hat k}(1), p_{\hat k}(-1)$ are all positive since if $p_k$ has two conjugate zeros, then $p_k(\lambda)$ is positive on the real line; if $p_k$ has two real zeros, by the assumption that the zeros are in $(-1, 1)$, $p_k(1), p_k(-1)$ are positive. This is a contradiction to the assumption that $p_{\hat k}(1) = 0$ or $p_{\hat k}(-1) = 0$ (as $p_{\hat k}(\lambda)$ is the convex combination of $p_k(\lambda)$ and $p_{k'}(\lambda)$).
Hence the zeros of $p_{\hat k}$ must be in the open unit disk for every $\delta \in [0,1]$. 
\end{proof}
\subsection{Connectedness properties of $\ca S_{\bf x}$ and $\ca S$}
The following topological property of the set $\ca S_{\bf x}$ has immediate algorithmic implications.
\begin{lemma}
  For the state feedback system, the set $\ca S_{\bf x}$ is connected and contractible in $\bb R^n$.
\end{lemma}
\begin{proof}
  The proof proceeds similar to the proof of Lemma~\ref{lemma:continuous_state_connected}. Putting 
  $\Gamma = \{\lambda \in \bb C: |\lambda | < 1\}^n$, it suffices to show that $\Gamma_{*}/S_n$ is connected and contractible.\footnote{As such, $\Gamma_*$ is the subset of $n$-dimensional complex-valued vectors with entries having modulus less than one and closed under conjugation.}  But this is immediate since any $v \in \Gamma_{*}/S_n$ is connected to $(0, \dots, 0)$ by the convex line segment $(1-\delta) v + \delta \, 0$ ($\delta \in (0,1)$) in $\Gamma_{*}$.
\end{proof}
For the output feedback case, the set $\ca S$ is not connected in general. Following is an example of a SISO system with more than one path-connected component.\footnote{
This example actually shows that the bound in Lemma~\ref{lemma:n/2_components_discrete} is tight.}
\begin{example} \label{ex.3}
Consider the LTI system $(A, b, c^\top)$ with,
  \begin{align*}
    A = \begin{pmatrix}
  0 & 1 & 0 & 0 \\
  0 & 0 & 1 & 0 \\
  0 & 0 & 0 & 1 \\
  0 & 0 & 0 & 0
\end{pmatrix}, \; b = \begin{pmatrix}
  0 \\
  0 \\
  0 \\
  1
\end{pmatrix}, \; c = \begin{pmatrix}
                     0.5184 \\
                     -2.448 \\
                     4.33  \\
                     -3.4 
                   \end{pmatrix}.
 \end{align*}
The feedback controllers are then parametrized by intervals in $\bb R$.
% The corresponding root locus
{Figure~\ref{fig:a}} %
depicts that 
the roots of the closed loop system are inside the unit disk for some interval, then become unstable, 
and subsequently reenter the unit disk as $k$ varies; as such, $\ca S$ has two connected components.
 \begin{figure}[ht]
   \centering
   \includegraphics[width=0.4\textwidth]{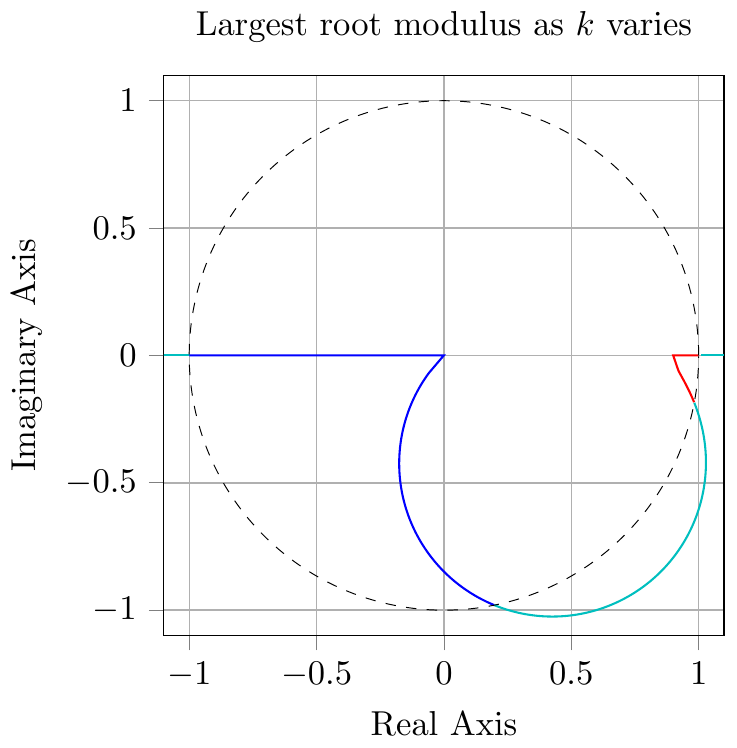}
   \caption{\label{fig:a}The figure depicts how the root with largest modulus varies with respect to the feedback gain $k$. The blue and red segments correspond to two stabilizing intervals.}
 \end{figure}
\end{example}
\subsubsection*{\bf Connected Components of $\ca S$ for SISO Systems}
We now show that there is at most $\ceil{\frac{n}{2}}$ connected components in $\ca S$.
 We will follow a similar line of reasoning as for the continuous systems presented in~\S~\ref{subsec:comonents_continuous}.
 % since the upper bound $n$ is straightforward to obtain. 
%Obtaining the upper bound of $\ceil{\frac{n}{2}}$, in the meantime, requires several technical details. 
 We shall demonstrate the upper bound $n$ first, followed by the upper bound of $\ceil{\frac{n}{2}}$.
The essential ideas for proving the two results are similar to the strategy we followed in Lemmas~\ref{lemma:2n_bound_continuous} and~\ref{lemma:n_bound_continuous}, with some subtle differences.
\begin{lemma}
  \label{lemma:n_components_discrete}
  The set $\ca S$ has at most $n$ connected components.
\end{lemma}
\begin{proof}
  Note that in this case, $\ca S$ is a subset of $\bb R$; it also suffices to assume that the system is in the controllable canonical form. 
  Consider  the characteristic polynomial of a closed-loop system, 
  \begin{align*}
    p(\lambda,k) = \lambda ^n + (a_{n-1} - k c_{n-1}) \lambda ^{n-1} + \dots + a_0 - k c_0,
  \end{align*}
  where $a = (a_{n-1,}, \dots, a_0)$ is the last row of $A$ and $c_j$'s are components of $c$.
Let $\zeta$ be the bijection between $\bb R^n$ and the set of monic $n$th degree polynomials, and 
  \begin{align*}
    \Gamma = \{a \in \bb R^n: \zeta(a) \text{ has roots on the unit disk in } \bb C\},
  \end{align*}
  and parameterize the line $\ell(k) \coloneqq a - k c$.
  Suppose $\ell(k) \cap \Gamma$ for finitely many $k$'s, listed in increasing order $\{k_1, \dots, k_l\}$ (this will be proven subsequently). Let $n_{p(x, k)}(\bb D)$ denote the number of roots of the closed loop characteristic polynomial on $\bb D$. Moreover, let $\gamma$ be a counterclockwise oriented unit circle in $\bb C$, tracing the boundary of $\bb D$. For each $k \in (k_j, k_l)$, we define,
\begin{align*}
  m(k) = \frac{1}{2 \pi i}\int_{\gamma} \frac{p'(\lambda ,k)}{p(\lambda, k)} \, d\lambda.
\end{align*}
Note that $p(\lambda , k)$ does not vanish on $\gamma$, and by Cauchy's Argument Principle~\cite{marden1966geometry},  $m(k)$ is the number of zeros of $p(\lambda , k)$ inside $\gamma$, that is, %Therefore, $m_r(k)$ is a well-defined integer-valued function and 
$m(k) = n_{p(\lambda ,k)}(\bb D)$. We further note that $m(k)$ is continuous in $k$, and as such,
$n_{p(\lambda ,k)}(\bb D)$ is constant on each interval $(k_j, k_{j+1})$. 
Hence, either $n_{p(\lambda ,k)}(\bb D) = n$ or $n_{p(\lambda ,k)}(\bb D)< n$,  respectively, corresponding to stabilizing and non-stabilizing gains $k$.

Now by inspecting the number of intersections between $\ell(k)$ and $\Gamma$, we can derive an upper bound on the number of connected components of $\ca S$.
  In this direction, when $\ell(k)$ intersects $\Gamma$ there is $\lambda _0 = e^{i \theta} \in \bb C$ such that,
   \begin{align*}
      \lambda _0^n  +(a_{n-1}-k c_{n-1} )\lambda _0^{n-1} + \dots + (a_0-k c_0) = 0,
   \end{align*}
   and therefore,     
      \begin{align*}
       k =  \frac{ \lambda _0^n + a_{n-1} \lambda _0^{n-1} + \dots + a_0 }{c_{n-1} \lambda_0^{n-1}  + c_{n-2}  \lambda_0^{n-2}+ \dots + c_0}.  \stepcounter{equation}\tag{\theequation}\label{myeq1}
  \end{align*}
  Note that $r(\lambda _0) = c_{n-1} \lambda _0^{n-1} + \dots + c_0 \neq 0$ (see details in the proof of Lemma~\ref{lemma:2n_bound_continuous}). This implies that $h( \lambda ) \coloneqq \text{\bf Im}(( \lambda _0^n + s( \lambda _0)) \overline{r( \lambda _0)}) = 0$, where $\overline{r( \lambda _0)}$ is the complex conjugate of $r( \lambda _0)$. Substituting $ \lambda _0 = e^{i \theta}$ into $( \lambda _0^n + s( \lambda _0)) \overline{r( \lambda _0)}$, we have
  \begin{align*}
  \alpha_n e^{i n \theta} + \alpha_{n-1} e^{i(n-1)\theta} &+ \dots + \alpha_0 \\
    & + \alpha_1 e^{-i \theta} + \dots + \alpha_{-(n-1)} e^{-i(n-1) \theta}  = 0,
  \end{align*}
  where $(\alpha_n, \dots, \alpha_1, \alpha_0, \alpha_{-1}, \dots, \alpha_{-(n-1)})$ are the corresponding coefficients when we expand the product.
 This implies that
 \begin{align*}
   \beta_n \sin(n\theta) + \dots + \beta_1 \sin(\theta) = 0,
 \end{align*}
 where $(\beta_n, \dots, \beta_1) \in \bb R^n$ are the corresponding real coefficients.
 We now note {Chebyshev polynomials of second kind} satisfy,
 \begin{align*}
   U_{n-1}(\cos(\theta)) \sin(\theta) = \sin(n \theta),
 \end{align*}
 where $U_{n-1}(\cos(\theta))$ is the {Chebyshev polynomial} of degree $n-1$ in $\cos(\theta)$. It 
 thus follows that,
 \begin{align}
   \label{eq:discrete_alg_poly}
 \nonumber  &\sin(\theta) \left(\beta_1 + \beta_2 U_1(\cos(\theta)) + \dots + \beta_n U_{n-1}(\cos(\theta))\right) \\
   &\eqqcolon \sin{\theta} \, g(\cos(\theta)),
 \end{align}
 where $g(\cos(\theta)) \in \bb R[\cos(\theta)]$ has degree $n-1$.
 By Fundamental Theorem of Algebra, there will be at most $n-1$ possible values for $\cos(\theta)$. Noting that $\theta = 0$ or $\theta = \pi$ also satisfy the above relation. Mind that for each value of $\cos(\theta)$, there are two possible $\theta$'s in $[0, 2 \pi)$, yielding a conjugate pair $e^{i \theta}$ and $e^{-i\theta}$. But this conjugate pair will be mapped to the same gain $k$ via~\eqref{myeq1}. Thereby, we will have at most $n+1$ possible values for such $k$'s. Since the set of stabilizing controllers is bounded, we have at most $n$ connected components. 
\end{proof}
\begin{remark}
  \label{remark:discrete_tangent}
  Figure~\ref{fig:discrete_tangent} depicts a similar situation as observed perviously for the continuous systems: two adjacent intervals $(k'', k_0)$ and $(k_0, k')$ could be both stabilizing but $k_0$ is only marginally stabilizing. The system parameters are,
{
\begin{align*}
&  A = \begin{pmatrix}
   \frac{2909}{1000} & -\frac{283}{100} & \frac{26129688223-120 \sqrt{6273911930105230}}{18036490000} \\
    1 & 0 & 0 \\
    0 & 1 & 0
  \end{pmatrix}, \\ &b = \begin{pmatrix}
    1 \\ 
    0 \\
    0
  \end{pmatrix}, \quad 
  c = \begin{pmatrix}
    0.1343 \\ -0.1846 \\ 0.0623
  \end{pmatrix}.
\end{align*}
}
%
  % To be more precise, there are two $k$'s with $k_1 = $ such that the root path intersects the unit disk.
\end{remark}
  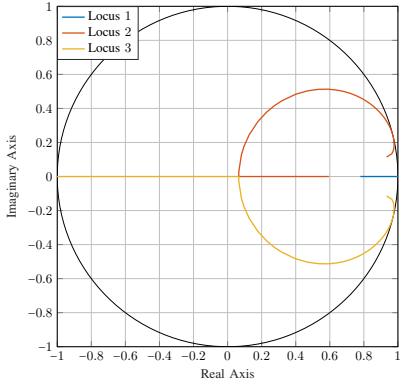
\begin{figure}[ht]
    \centering
    \input{discrete_tangent_plot.tex}
    \caption{\label{fig:discrete_tangent}  Two adjacent stabilizing intervals $(k_{j-1}, k_j)$ and $(k_j, k_{j+1})$ where $k_j$ is not stabilizing.}
  \end{figure}
Now we prove that the bound on the number of connected components of $\ca S$ can be tightened to $\ceil{\frac{n}{2}}$. The strategy for proving this bound is similar to Lemma~\ref{lemma:n_bound_continuous}. That is, we need to examine the implications of having two adjacent stabilizing intervals $(k_{j-1}, k_j)$ and $(k_j, k_{j+1})$; see Figure~\ref{fig:discrete_tangent}.

We first make an observation.
\begin{proposition}
  \label{prop:poly_derivative_discrete}
  Let
  \begin{align*}
    \varphi(\lambda) = a_{-n} \lambda^{-n} + a_{n-1} \lambda^{-(n-1)} + \dots + a_0 + a_1 \lambda + \dots + a_n \lambda^n ,
  \end{align*}
  and
  \begin{align*}
    h(\theta) = \text{\bf Im} (\varphi(\lambda)|_{\lambda=e^{i\theta}}),
  \end{align*}
where $ \varphi(\lambda) \in \bb R[\lambda]$. Then $h'(\theta) = \text{\bf Re}\left( (\lambda \varphi'(\lambda))\big \vert_{\lambda=e^{i\theta}}\right)$; note that $h'(\theta)$ denotes differentiation with respect to $\theta$ and $\varphi'$ refers to differentiation with respect to $\lambda$.
\end{proposition}
\begin{proof}
  By linearity of the operations involved, it suffices to show this relation holds for $\lambda^{-l}$ and $\lambda^l$, when $l \in \bb N$. For $\lambda^j$, $v(\theta) \coloneqq \text{\bf Im}(\lambda^l|_{\lambda=e^{i \theta}}) = \sin(l\theta)$ and $v'(\theta)= l \cos(l\theta)$; however, 
  \begin{align*}
    \text{\bf Re}\left( (\lambda (\lambda^l)')|_{\lambda=e^{i\theta}}\right) = \text{\bf Re}\left( l \lambda^l|_{\lambda=e^{i\theta}}\right) = l \cos(l \theta).
  \end{align*}
  For $\lambda^{-l}$, $w(\theta) \coloneqq \text{\bf Im}(\lambda^{-l}|_{e^{i \theta}}) = -  \sin(l \theta)$ and $w'(\theta) = - l \cos(l \theta)$. The proof now follows by observing that, 
  \begin{align*}
    \text{\bf Re}\left( (\lambda (\lambda^{-l})')|_{\lambda=e^{i\theta}}\right) = \text{\bf Re}\left( -l \lambda^{-l}|_{\lambda=e^{i\theta}}\right) = -l \cos(l \theta).
  \end{align*}
\end{proof}
We are now in the position to prove the tight bound of $\ceil{n/2}$ on the number of connected components of $\ca S$.
\begin{lemma}
  \label{lemma:n/2_components_discrete}
  When $\ca S \neq \emptyset$, it has at most $\ceil{\frac{n}{2}}$ connected components.
\end{lemma}
\begin{proof}
  Similar to the continuous case, we need to explore the implications of having two adjacent intervals be stabilizing. Suppose $(k_{j-1}, k_j)$ and $(k_j, k_{j+1})$ are two such intervals, where $k$'s in each interval are stabilizing. By assumption, $p(\lambda_0, k_j) = 0$ for some $|\lambda_0| = 1$.  By Proposition~\ref{prop:asymptotic_root}, $\lambda_0 = e^{i \theta_0}$ is a simple root for $p(\lambda, k_j) \in \bb R[\lambda]$ since the unit disk has positive Gaussian curvature (see argument in the proof to Lemma~\ref{lemma:n_bound_continuous}).
  Let us denote by $s(\lambda) = a_{n-1}\lambda^{n-1} + \dots + a_0$ and $r(\lambda)= c_{n-1}\lambda^{n-1} + \dots + c_0$. Now the curve $\eta(t)$ is differentiable at $\lambda_0$ with
 \begin{align*}
   \frac{d{\eta}}{d t}\Big \vert_{t=k_j} = \frac{r(\lambda_0)}{(\lambda^n + s(\lambda))'|_{\lambda=\lambda_0}},
 \end{align*}
 by appealing to the asymptotic formula or just formally differentiating.
 Geometrically, at $k_j$, the derivative should be orthogonal to $\lambda_0$ if the curve $\eta(k)$ is tangent to the unit sphere at $\lambda_0$. This implies that,
 \begin{align*}
   \frac{r(\lambda)}{\left(\lambda^n + s(\lambda)\right)'}\Big\vert_{\lambda=\lambda_0} = i \gamma \lambda_0,
 \end{align*}
 where $\gamma \in \bb R$. Hence, 
 \begin{align}
   \label{eq:alg_derivative_discrete}
   \text{\bf Re}\left(r( \bar{\lambda} ) \left(\lambda^n + s(\lambda)\right)' \lambda|_{\lambda = e^{i \theta_0}}\right) = 0.
 \end{align}
 Recall that $k_j$'s correspond to the roots of $\sin(\theta) g (\cos(\theta))$ via~\eqref{myeq1}. So if $(k_{j-1}, k_j)$ and $(k_j, k_{j+1})$ are both stabilizing, then the $\theta_0$ that maps to $k_j$ must satisfy, 
 \begin{align}
   &\text{\bf Im}\left( \left(\lambda^n + s(\lambda)\right) r(\bar{\lambda}) |_{\lambda= e^{i \theta_0}}\right) = 0, \label{myeq2}\\
   &\text{\bf Re}\left(r(\bar{\lambda}) p_{k_j}'(\lambda) \lambda|_{\lambda= e^{i \theta_0}}\right) = 0. \label{myeq3}
 \end{align}
 Let us now show that these two relations imply if $\cos(\theta_0)$ is a solution to $g(\cos(\theta_0))$ (recall that~\eqref{myeq2} is equivalent to $\sin(\theta) g(\cos(\theta)) = 0$), then $g'(\cos(\theta_0)) = 0$ due to~\eqref{myeq3}. Note that $\theta_0 \neq 0 \text{ or } \pi$ by~\eqref{eq:alg_derivative_discrete}.  
 
 Let us consider the function $G(\theta)$ defined by
 \begin{align*}
   G(\theta) = \text{\bf Im} \left(\left(\lambda^n + s(\lambda)\right) r(\bar{\lambda})|_{\lambda=e^{i\theta}} \right).
 \end{align*}
 Note that $G(\theta) = \sin(\theta) g(\cos(\theta))$.
 We adopt the notation $G'(\theta)$ to refer to differentiation with respect to $\theta$.
 Let us define $q(\lambda) = r(1/\lambda)$ and observe that $q'(\lambda) = -r'(1/\lambda)/\lambda^2$. Then
 \begin{align*}
   G(\theta_0) = \text{\bf Im} \left(\left( \lambda^n + s(\lambda) \right) q(\lambda)|_{\lambda=e^{i\theta_0}} \right).
 \end{align*}
 By Proposition~\ref{prop:poly_derivative_discrete}, we have
 \begin{align*}
   G'(\theta_0) &= \text{\bf Re} \left(\lambda\left((\lambda^n + s(\lambda)) q(\lambda)\right)'|_{\lambda=e^{i\theta_0}} \right) \\
              &= \text{\bf Re}\left( \lambda \left(\lambda^n + s(\lambda)\right)' q(\lambda)|_{\lambda=e^{i \theta_0}}\right.\\
&\quad  \left. + \lambda \left(\lambda^n + s(\lambda)\right)q'(\lambda)|_{\lambda=e^{i \theta_0}}\right).
 \end{align*}
 Noting that $\lambda_0^n + s(\lambda_0) - k_j r(\lambda_0) = 0$, i.e., $\lambda_0^n + s(\lambda_0) = k_j r(\lambda_0)$, it follows that,
  \begin{align*}
   G'(\theta_0) &= \text{\bf Re} \left( \left( \lambda(\lambda^n + s(\lambda))' r(\frac 1 \lambda)
   + \lambda k_j r(\lambda) q'(\lambda) \right)|_{\lambda=e^{i\theta_0}} \right) \\
    &= \text{\bf Re}\left( \big( \lambda(\lambda^n + s(\lambda))' r(\frac 1 \lambda) - \frac{1}{\lambda} k_j r(\lambda) r'(\frac 1 \lambda) \big) |_{\lambda=e^{i\theta_0}} \right) \\
   &= \text{\bf Re}\left( \left( \lambda \left(\lambda^n + s(\lambda)\right)' r(\bar{\lambda}) -  k_jr(\lambda)\bar{\lambda} r'(\bar{\lambda}) \right)|_{\lambda=e^{i\theta_0}} \right)  \\
   &= \text{\bf Re}\left( \left( \lambda \left(\lambda^n + s(\lambda)\right)' r(\bar{\lambda}) -  k_jr(\bar{\lambda})\lambda r'(\lambda) \right)|_{\lambda=e^{i\theta_0}} \right)  \\
                &= \text{\bf Re}\left( r(\bar{\lambda}) p_{k_j}'(\lambda) \lambda|_{\lambda=e^{i \theta_0}}\right).
\end{align*}
The above identity is precisely~\eqref{myeq3}, i.e., $G'(\theta_0) = 0$. But by the chain rule,
\begin{align*}
  G'(\theta_0) = \cos(\theta_0) g(\cos(\theta_0)) - \sin^{2}(\theta_0) g'(\cos(\theta_0)) = 0.
\end{align*}
 Since $\sin(\theta_0) \neq 0$, it follows that $g'(\cos(\theta_0)) = 0$. Putting $v=\cos(\theta_0)$, 
 it now follows that $v$ must be a multiple root if the two adjacent intervals are both stabilizing.  
 
 Now let $\{u_1, \dots, u_\mu, v_1, \dots, v_\nu\}$ denote the roots of the polynomial $g(v) = g(\cos(\theta)) = 0$, where $v_j$'s are roots of multiplicity greater than $1$. We have $\mu + 2 \nu \le n-1$. These roots will be mapped to $k_j$'s via~\eqref{myeq1}. Let $\Lambda = \{k_j\}_{j=1}^{\mu'}$ and $\Pi = \{k_j\}_{j=1}^{\nu'}$ be the corresponding sets to $\{u_j\}$ and $\{v_j\}$. Since $\theta = 0$ or $\theta = \pi$ also corresponds to real $k$'s, we add these two to the set $\Lambda$. Now $\mu' + 2 \nu' \le n+1$. Let $\Upsilon$ be the collection of nonstabilizing intervals. Note that if $k_i \in \Lambda$, then one of the intervals $(k_{j-1}, k_j)$ and $(k_j, k_{j+1})$ must be not stabilizing. Hence, $k_j$ must be the endpoint of some interval in $\Upsilon$. It now follows that the cardinality of $\Lambda$ is given by
\begin{align*}
  \mu' = |\Lambda| \le 2 + 2(| \Upsilon | - 2).
\end{align*}
On the other hand, the total number of stabilizing intervals would be $\mu' + \nu' + 1 - |\Upsilon|$ and thereby,
\begin{align*}
  \mu' + \nu' + 1 - |\Upsilon| \le \mu' + \nu' + 1 - \frac{\mu' + 2}{2} \le \frac{n+1}{2}.
\end{align*}
Consequently, the total number of stabilizing intervals is at most $\ceil{\frac {n} 2}$.\footnote{In particular, when $n=2$, the set $\ca S$ corresponding to the triplet $(A, b, c^\top)$, if nonempty, is connected.}
\end{proof}
\begin{remark}
  As shown in Example~\ref{ex.3}, this bound is tight. \par
\end{remark}
{
We now observe the property that the stabilizing and non-stabilizing intervals in $\ca S$ interlace is also generic. The proof is almost verbatim to the one to Lemma~\ref{lemma:interlacing} except that we need the nonemptyness of $\ca U$ for discrete systems; this is immediate from Example~\ref{ex.3}.
\begin{lemma}
    For a controllable and observable system $(A, b, c^\top)$, the property that stabilizing and non-stabilizing intervals in $\ca S$ interlace is generic.
\end{lemma}
}
%When $n=2$, the set of stabilizing controllers for $(A, b, c^T)$ is connected.
%\begin{corollary}
%\sout{  If $n=2$, the set $\ca S$ corresponding to the triplet $(A, b, c^\top)$ is connected.}
%\end{corollary}
%\begin{proof}
%  \sout{If $\ca S = \emptyset$, no further steps are necessary. Otherwise, the conclusion follows from Lemma~\ref{lemma:n/2_components_discrete}.}
%\end{proof}
\subsubsection*{\bf An algorithm for characterizing the connected components of $\ca S$}
%\paragraph{An Algorithm}
Similar to the continuous case, our analysis leads to an algorithm for identifying the stabilizing intervals
for discrete-time linear SISO systems. We summarize this algorithm below.
\begin{algorithm}[H]
\floatname{algorithm}{Algorithm $2$:}
\renewcommand{\thealgorithm}{}
\caption{\bf Identifying stabilizing intervals of $\ca S$}
\label{discrete_alg}
\begin{algorithmic}[1]
\STATE Find real zeros $\{\lambda_1, \dots, \lambda_l\}$ of the real polynomial~\eqref{eq:discrete_alg_poly}. Note that these zeros correspond to $\cos(\theta)$. Mapping these zeros to the corresponding values of $\lambda$, we get $\{\lambda_1, \bar{\lambda}_1, \dots, \lambda_l, \bar{\lambda_l}\}$. Appending $\{-1, 1\}$ to this list if necessary, we get a new list $L = \{1, -1, \lambda_1, \bar{\lambda}_1, \dots, \lambda_l, \bar{\lambda}_l\}$. Mapping $L$ to $\{k_1, \dots, k_{l'}\}$ (order this list in an  increasing manner) by~\eqref{myeq1}.
\STATE Start from the interval $(-\infty, k_1)$. It is not stabilizing by Proposition~\ref{prop:discrete_bounded}. Check the multiplicity of $\lambda_{1'}$ that maps to $k_1$. If $\lambda_{1'}$ is simple, then $(k_1, k_2)$ is not stabilizing; if $\lambda_{1'}$ is not simple, check whether~\eqref{eq:alg_derivative_discrete} is satisfied; if not, i.e., the corresponding derivative is not pure imaginary, then $(k_1, k_2)$ is not stabilizing. On the other hand, if this derivative is pure imaginary, then $(k_1, k_2)$ is stabilizing. Continue the process.
\end{algorithmic}
\end{algorithm}
{We demonstrate the progression of Algorithm 2 on the example in Remark~\ref{remark:discrete_tangent}.
%\begin{enumerate}
%\item 
The characteristic polynomial of the closed-loop system is given by
  \begin{align*}
    p(\lambda, k) &=
    \lambda^3 +2.909 \lambda^2 + 2.83 \lambda - 0.9217272705 \\
    &\quad +k (0.1343\lambda^2 -0.1846\lambda + 0.06229).
  \end{align*}
  Furthermore, 
  \begin{align*}
    g(\cos(\theta)) &=0.24916\, \left( \cos \left( \theta \right)  \right) ^{2}-
 0.4840272752\,\cos \left( \theta \right)\\
 &\quad + 0.2350722459,
  \end{align*}
has a root $\frac{\sqrt{6018}}{40}$ with multiplicity $2$.
Then $\sin(\theta_1) =0$ is mapped to $k_1=0.06065638$ and $k_3 = 20.09687366$; $\cos(\theta_2) =\frac{\sqrt{6018}}{40}$ on the other hand maps to $k_2 = 0.6198635016$.
%\item 
By Corollary~\ref{cor:bounded_discrete}, the two unbounded intervals $(-\infty, k_1)$ and $(k_3, \infty)$ are both non-stabilizing. We note that $k_1$ corresponds to $\sin(\theta)=0$ and
% by our analysis, 
%it corresponds interlacing stabilizing and non-stabilizing intervals. 
thereby $(k_1, k_2)$ is stabilizing. 
  %\item 
  Since $k_2$ is obtained from a multiple root of $g(\cos(\theta))$, we need to check condition~\eqref{eq:alg_derivative_discrete}. In this case, we conclude that the interval $(k_2, k_3)$ is stabilizing.\footnote{Alternatively, we note that $k_3$ is obtained from $\sin(\theta)=0$, and as such, $(k_2, k_3)$ must be stabilizing since $(k_3, \infty)$ is not.}
%\end{enumerate}
} 
%
%%%%%%%%%%%%%%%%%%%%%%%%%%%%%%%%%%%%%%%%%%%%%%%%%%%%%%%%%%%%%%%%%%%%%%%%%%%%%%%%%%%%%%%
%\balance
\section{Conclusion}
The motivation for this work stems from recent interest in devising learning type algorithms for
control synthesis, that evolve over the set of stabilizing feedback gains. 
This in turn, has inspired the need to further examine the topological properties of these sets.
We envisage that some of these properties might been observed in the earlier literature in system theory and
known to experts;\footnote{Particularly those with deep affinity for root locus and geometry of polynomials.}
however, this work is an attempt to gather and prove these
properties in a concise and rigorous manner using basic topology and the theory of polynomials.
In this work, we have focused on topological and metrical properties of 
stabilizing state feedback gains and SISO output feedback gains for continuous and discrete time linear systems;
some of these results have MIMO counterparts that are discussed in~\cite{MIMO2019}. In this latter case,
topological arguments turn out to be even more dominant for characterizing the set of stabilizing feedback gains, with less reliance on the geometry of polynomials.
\section*{Acknowledgments}
The first author is grateful to Dillon Foight for providing a preliminary version of Figures~\ref{fig:continuous_tangent} and~\ref{fig:discrete_tangent} and many helpful remarks. The authors also acknowledge their discussions with Maryam Fazel, Sham Kakade, and Rong Ge, on exploring connections between control theory and learning, and Han Feng and Javad Lavaei for discussions on connectedness properties of stabilizing gains. This research was supported by DARPA Lagrange Grant FA8650-18-2-7836.
\ifCLASSOPTIONcaptionsoff
  \newpage
\fi
\balance 
\bibliographystyle{ieeetran}
\bibliography{ref}
\end{document}

%% file: pkgs.tex
\usepackage[utf8]{inputenc}
\usepackage{xcolor}
\usepackage{tkz-graph}

\usepackage{cancel}
\usepackage[normalem]{ulem}
\usepackage[T1]{fontenc} % 8-bit encoding
\usepackage{color}
\usepackage[leqno]{amsmath}
\usepackage{amsfonts, amsthm, amssymb, commath} % Math packages
\usepackage{mathtools}
\usepackage{graphicx}
\usepackage{mathrsfs}
\usepackage{caption}
\usepackage{subcaption}
\usepackage{listings}
\usepackage[]{eufrak}
\usepackage[]{hyperref}
\usepackage[]{float}
\usepackage{tikz}
\usetikzlibrary{calc, arrows,matrix,positioning,scopes}
\usetikzlibrary{shapes.geometric}
\usetikzlibrary{decorations.markings}
\usepackage{tikz-cd}
\usetikzlibrary{cd}
\usepackage{booktabs}
\usepackage{enumerate}
\usepackage{eufrak}
\usepackage{MnSymbol}
\usepackage{fancyhdr}
\usepackage{tikz}
\usepackage{algorithm}
\usepackage{algorithmic}
\usetikzlibrary{matrix,decorations.pathreplacing,calc}
\pgfkeys{tikz/mymatrixenv/.style={decoration=brace,every left delimiter/.style={xshift=3pt},every right delimiter/.style={xshift=-3pt}}}
\pgfkeys{tikz/mymatrix/.style={matrix of math nodes,left delimiter={(}, right delimiter={)}, inner sep=1pt,column sep=1.5em,row sep=0.5em,nodes={inner sep=0pt}}}
\pgfkeys{tikz/mymatrixbrace/.style={decorate,thick}}

\usepackage{cuted}
\usetikzlibrary{shapes.misc}
  \usepackage{pgfplots}
  \pgfplotsset{compat=newest}
  \usetikzlibrary{plotmarks}
  \usetikzlibrary{arrows.meta}
  \usepgfplotslibrary{patchplots}
  \usepackage{grffile}
  \usepackage{amsmath}

  \pgfplotsset{plot coordinates/math parser=false}
  \newlength\figureheight
  \newlength\figurewidth
\tikzset{cross/.style={cross out, draw=black, minimum size=2*(#1-\pgflinewidth), inner sep=0pt, outer sep=0pt},
cross/.default={1pt}}

%% file: defs.tex
\newcommand{\ca}[1]{\mathcal{#1}}

\newcommand{\bb}[1]{\mathbb{#1}}

\makeatletter
\newcommand{\leqnomode}{\tagsleft@true}
\newcommand{\reqnomode}{\tagsleft@false}
\makeatother

\DeclarePairedDelimiter{\ceil}{\lceil}{\rceil}

\newtheorem{theorem}{Theorem}[section]
\newtheorem{corollary}{Corollary}[theorem]
\newtheorem{lemma}[theorem]{Lemma}
\newtheorem{proposition}[theorem]{Proposition}

\newtheorem{observation}[theorem]{Observation}
\newtheorem{remark}[theorem]{Remark}

\newtheorem{example}{Example}

%% file: continuous_tangent.tex
\definecolor{mycolor1}{rgb}{0.00000,0.44700,0.74100}%
\definecolor{mycolor2}{rgb}{0.85000,0.32500,0.09800}%
\definecolor{mycolor3}{rgb}{0.92900,0.69400,0.12500}%
\begin{tikzpicture}[scale=0.5]

\begin{axis}[%
width=4.602in,
height=3.566in,
at={(0.772in,0.481in)},
scale only axis,
unbounded coords=jump,
xmin=-15,
xmax=0,
ymin=-7.5,
ymax=7.5,
xlabel={Real Axis},
ylabel={Imaginary Axis},
axis background/.style={fill=white},
xmajorgrids,
ymajorgrids,
legend style={at={(0,1)},anchor=north west, legend cell align=left, align=left, draw=white!15!black}
]
\draw (-0.354698225677819,1.00108961713154) node[cross=2pt] {};
\draw (-0.354698225677819,-1.00108961713154) node[cross=2pt] {};
\addplot [color=mycolor1, line width=1.0pt]
  table[row sep=crcr]{%
-0.354698225677819	1.00108961713154\\
-0.349465822083785	1.00344337711179\\
-0.34812124710757	1.00406518965821\\
-0.346432736946199	1.00485600360314\\
-0.34431338636297	1.00586439241268\\
-0.341655103844884	1.00715429418056\\
-0.338324044549953	1.00881054060378\\
-0.33415555886618	1.01094661570718\\
-0.328949082507148	1.01371554940476\\
-0.322463926958977	1.01732510222237\\
-0.314417868041968	1.02205851642012\\
-0.304491927592826	1.02830174169066\\
-0.292346739292493	1.03657647007966\\
-0.277657561478706	1.04757443750427\\
-0.260173849480896	1.06218154191871\\
-0.239800957739291	1.08147252280899\\
-0.216683560061022	1.106656689625\\
-0.191251872307521	1.13897446323483\\
-0.164196473459144	1.17958085903414\\
-0.136377755408117	1.2294735156734\\
-0.108720571531596	1.28949868353332\\
-0.0821483445899266	1.36041881202589\\
-0.0575764923267429	1.44299875451198\\
-0.0359526748169604	1.53807738618155\\
-0.0183226236284267	1.64661359821359\\
-0.00590734444599295	1.76970967278811\\
-0.000187002284114145	1.90861841344252\\
-0.00299312239984881	2.06473826646529\\
-0.0166139107685004	2.2395971069357\\
-0.0439190324151744	2.43482184427192\\
-0.0885112547302268	2.65208749052824\\
-0.154913585246009	2.89303521632065\\
-0.248802198823774	3.15914331111727\\
-0.377297714248701	3.45152662928323\\
-0.549330375221843	3.77062709525877\\
-0.776098569700807	4.11573678771984\\
-1.07164508587593	4.48425973292062\\
-1.45358181597388	4.87055626635237\\
-1.94400161765963	5.2640977440626\\
-2.57062615007686	5.64642644531347\\
-3.36825123925588	5.9859012169623\\
-4.38056733207016	6.22792078964176\\
-5.0205837944799	6.28589410683153\\
-5.66245264523166	6.27446976450313\\
-6.47179413367598	6.16117533375336\\
-7.28286166698287	5.93188309910392\\
-8.30486851780994	5.45605286851885\\
-8.81649321978989	5.12385049540644\\
-9.3284629470838	4.71200574521551\\
-9.85045963684721	4.18619037480752\\
-10.372735028362	3.50563611783465\\
-10.8952579850304	2.54935495240839\\
-11.1566038407559	1.84277662368324\\
-11.2872965326178	1.33868797372205\\
-11.4180019306217	0.391984583896645\\
-11.4301199097003	3.38636047360415e-07\\
-11.0498732957531	0\\
-9.39845510703528	0\\
-7.32635505455467	0\\
-6.56768637859538	0\\
-6.12923999450568	0\\
-5.8407167694877	0\\
-5.63824762642569	0\\
-5.03063684362512	0\\
-5	0\\
};
\draw (-5, 0) circle (2pt);
\addlegendentry{Locus 1}

\addplot [color=mycolor2, line width=1.0pt]
  table[row sep=crcr]{%
-0.354698225677819	-1.00108961713154\\
-0.349465822083785	-1.00344337711179\\
-0.34812124710757	-1.00406518965821\\
-0.346432736946199	-1.00485600360314\\
-0.34431338636297	-1.00586439241268\\
-0.341655103844884	-1.00715429418056\\
-0.338324044549953	-1.00881054060378\\
-0.33415555886618	-1.01094661570718\\
-0.328949082507148	-1.01371554940476\\
-0.322463926958977	-1.01732510222237\\
-0.314417868041968	-1.02205851642012\\
-0.304491927592826	-1.02830174169066\\
-0.292346739292493	-1.03657647007966\\
-0.277657561478706	-1.04757443750427\\
-0.260173849480896	-1.06218154191871\\
-0.239800957739291	-1.08147252280899\\
-0.216683560061022	-1.106656689625\\
-0.191251872307521	-1.13897446323483\\
-0.164196473459144	-1.17958085903414\\
-0.136377755408117	-1.2294735156734\\
-0.108720571531596	-1.28949868353332\\
-0.0821483445899266	-1.36041881202589\\
-0.0575764923267429	-1.44299875451198\\
-0.0359526748169604	-1.53807738618155\\
-0.0183226236284267	-1.64661359821359\\
-0.00590734444599295	-1.76970967278811\\
-0.000187002284114145	-1.90861841344252\\
-0.00299312239984881	-2.06473826646529\\
-0.0166139107685004	-2.2395971069357\\
-0.0439190324151744	-2.43482184427192\\
-0.0885112547302268	-2.65208749052824\\
-0.154913585246009	-2.89303521632065\\
-0.248802198823774	-3.15914331111727\\
-0.377297714248701	-3.45152662928323\\
-0.549330375221843	-3.77062709525877\\
-0.776098569700807	-4.11573678771984\\
-1.07164508587593	-4.48425973292062\\
-1.45358181597388	-4.87055626635237\\
-1.94400161765963	-5.2640977440626\\
-2.57062615007686	-5.64642644531347\\
-3.36825123925588	-5.9859012169623\\
-4.38056733207016	-6.22792078964176\\
-5.0205837944799	-6.28589410683153\\
-5.66245264523166	-6.27446976450313\\
-6.47179413367598	-6.16117533375336\\
-7.28286166698287	-5.93188309910392\\
-8.30486851780994	-5.45605286851885\\
-8.81649321978989	-5.12385049540644\\
-9.3284629470838	-4.71200574521551\\
-9.85045963684721	-4.18619037480752\\
-10.372735028362	-3.50563611783465\\
-10.8952579850304	-2.54935495240839\\
-11.1566038407559	-1.84277662368324\\
-11.2872965326178	-1.33868797372205\\
-11.4180019306217	-0.391984583896645\\
-11.4301199097003	-3.38636047360415e-07\\
-11.8346026938872	0\\
-14.4179833187548	0\\
% -22.991948710781	0\\
% -31.9395011346396	0\\
% -42.6873830812112	0\\
% -55.9512396060725	0\\
% -72.4809508612447	0\\
% -1442.91704642187	0\\
inf	0\\
};
\addlegendentry{Locus 2}
\draw (-0.115603548644362,0) node[cross=2pt] {};
\addplot [color=mycolor3, line width=1.0pt]
  table[row sep=crcr]{%
-0.115603548644362	0\\
-0.127117916683786	0\\
-0.130077622608547	0\\
-0.133794942867554	0\\
-0.138461666544935	0\\
-0.144316589806698	0\\
-0.151655844695855	0\\
-0.16084450467152	0\\
-0.172328694429981	0\\
-0.186646386230695	0\\
-0.20443321279309	0\\
-0.226416664720044	0\\
-0.253388089086732	0\\
-0.286138613423208	0\\
-0.325347484343209	0\\
-0.371428075650485	0\\
-0.424372886486915	0\\
-0.483675986389851	0\\
-0.548402102053118	0\\
-0.61739127456786	0\\
-0.68949919158586	0\\
-0.763766238680869	0\\
-0.839477522040414	0\\
-0.916141330821596	0\\
-0.993431633838254	0\\
-1.07112694501433	0\\
-1.14905986699725	0\\
-1.22708024460308	0\\
-1.3050301159589	0\\
-1.38272758989444	0\\
-1.45995716803564	0\\
-1.53646476448203	0\\
-1.61195634365015	0\\
-1.68609959618249	0\\
-1.75852840873784	0\\
-1.82885008214203	0\\
-1.896655331349	0\\
-1.96153106226308	0\\
-2.02307575551209	0\\
-2.08091698614534	0\\
-2.13473017956157	0\\
-2.1842571983206	0\\
-2.20864242952288	0\\
-2.22932288404116	0\\
-2.25131119425524	0\\
-2.26984741474419	0\\
-2.28943751981162	0\\
-2.29799001921255	0\\
-2.30585246798549	0\\
-2.31323960173003	0\\
-2.3200693319717	0\\
-2.32640393190619	0\\
-2.32940247709091	0\\
-2.3308622216851	0\\
-2.33229655399501	0\\
-2.33242826392119	0\\
-2.3325597617647	0\\
-2.33745982075609	0\\
-2.36488212023466	0\\
-2.38840687426255	0\\
-2.40837585500095	0\\
-2.4251627343302	0\\
-2.43915172066292	0\\
-2.49633395939439	0\\
-2.5	0\\
};
\draw (-2.5,0) circle (2pt);
\addlegendentry{Locus 3}

\end{axis}
\end{tikzpicture}%

%% file: discrete_tangent_plot.tex
\definecolor{mycolor1}{rgb}{0.00000,0.44700,0.74100}%
\definecolor{mycolor2}{rgb}{0.85000,0.32500,0.09800}%
\definecolor{mycolor3}{rgb}{0.92900,0.69400,0.12500}%
\begin{tikzpicture}[scale=0.5]

\begin{axis}[%
width=3.566in,
height=3.566in,
at={(1.29in,0.481in)},
scale only axis,
unbounded coords=jump,
xmin=-1,
xmax=1,
xlabel style={font=\color{white!15!black}},
xlabel={Real Axis},
ymin=-1,
ymax=1,
ylabel style={font=\color{white!15!black}},
ylabel={Imaginary Axis},
axis background/.style={fill=white},
xmajorgrids,
ymajorgrids,
legend style={at={(0,1)}, anchor=north west, legend cell align=left, align=left, draw=white!15!black}
]
\draw[thin] (0,0) circle (1);
\addplot [color=mycolor1, line width=1.0pt]
  table[row sep=crcr]{%
1.04154334547393	0\\
0.967017401137887	0\\
0.960010636907936	0\\
0.953054362290682	0\\
0.946244465458992	0\\
0.939641524695622	0\\
0.933275423909701	0\\
0.927153840128371	0\\
0.921270558515566	0\\
0.915611871667917	0\\
0.910160864125743	0\\
0.904900011866762	0\\
0.899812625992229	0\\
0.894883570519009	0\\
0.890099551723503	0\\
0.885449168070167	0\\
0.880922834638742	0\\
0.876512648022327	0\\
0.872212228481013	0\\
0.868016558915025	0\\
0.863921830306457	0\\
0.859925297709989	0\\
0.856025147845801	0\\
0.852220377784086	0\\
0.848510683478651	0\\
0.844896356641867	0\\
0.841378188443436	0\\
0.837957378643533	0\\
0.834635448972145	0\\
0.831414159805186	0\\
0.828295429447565	0\\
0.82528125560463	0\\
0.822373638901078	0\\
0.819574508586965	0\\
0.816885650847125	0\\
0.814308640395486	0\\
0.811844776276553	0\\
0.809495022998077	0\\
0.807259958263405	0\\
0.805139728640913	0\\
0.80313401448407	0\\
0.802146182134342	0\\
0.801242005287223	0\\
0.80031220959776	0\\
0.79946238642399	0\\
0.798987496969838	0\\
0.798758817846852	0\\
0.798535688326688	0\\
0.798521486811721	0\\
0.798507307453084	0\\
0.797793337873377	0\\
0.796232545110388	0\\
0.794777221851749	0\\
0.793424143841673	0\\
0.792169692379938	0\\
0.791009905878358	0\\
0.789940537418569	0\\
0.788957116100157	0\\
0.780415600996786	0\\
0.779592810850193	0\\
};
\addlegendentry{Locus $1$}

\addplot [color=mycolor2, line width=1.0pt]
  table[row sep=crcr]{%
0.933728327263038	0.114517943502045\\
0.963391601088158	0.133235285191056\\
0.965950158159915	0.13713479630936\\
0.968366005978441	0.141585988123832\\
0.970576596805151	0.146568885184212\\
0.972535222261772	0.152057699809301\\
0.974208479851332	0.158028192242201\\
0.975571775429015	0.16446177680834\\
0.976604880358111	0.171346911507384\\
0.977288411053443	0.178678859981665\\
0.977601326033031	0.186458719640907\\
0.977519220635623	0.194692243589618\\
0.977013149224172	0.203388699602975\\
0.976048753650409	0.212559844007609\\
0.974585542210078	0.222219007989692\\
0.972576216648902	0.232380260761747\\
0.969965981353522	0.243057603125468\\
0.966691791725106	0.254264142637042\\
0.962681512081745	0.266011201169231\\
0.957852960646106	0.278307304121314\\
0.952112822558047	0.291156996428389\\
0.945355412866458	0.304559423015273\\
0.937461270978256	0.318506599605753\\
0.928295566587166	0.332981282718821\\
0.917706294955367	0.34795432357036\\
0.905522236725693	0.36338135685302\\
0.89155065426707	0.379198627970812\\
0.875574692916936	0.395317694968563\\
0.857350451364069	0.411618644100139\\
0.836603680776768	0.42794131436639\\
0.813026067070305	0.444073808942711\\
0.786271044859199	0.459737232449454\\
0.755949085078022	0.474565046231198\\
0.72162239089248	0.488074516269009\\
0.682798928264884	0.499626115944091\\
0.638925708276658	0.508363753321968\\
0.589381227926346	0.513122761203826\\
0.533466964480966	0.512279840710781\\
0.470397805411754	0.503488790673311\\
0.399291281323482	0.483163037020644\\
0.319155452897268	0.445292406129998\\
0.274036282569746	0.415536998471938\\
0.228875284490919	0.377961352325681\\
0.178056293588025	0.322985149579327\\
0.127197316427285	0.245934224926801\\
0.096727541204901	0.178047907094137\\
0.0814882707916639	0.128833386681334\\
0.0662462255770164	0.0319480732142694\\
0.065263371112114	0\\
0.0962593796726613	0\\
0.252346939653836	0\\
0.354921625310215	0\\
0.407505576847541	0\\
0.442377791625308	0\\
0.467869561181753	0\\
0.487521118786594	0\\
0.503188887487703	0\\
0.515973195268751	0\\
0.589952080386449	0\\
0.59494181312598	0\\
};
\addlegendentry{Locus $2$}

\addplot [color=mycolor3, line width=1.0pt]
  table[row sep=crcr]{%
0.933728327263038	-0.114517943502045\\
0.963391601088158	-0.133235285191056\\
0.965950158159915	-0.13713479630936\\
0.968366005978441	-0.141585988123832\\
0.970576596805151	-0.146568885184212\\
0.972535222261772	-0.152057699809301\\
0.974208479851332	-0.158028192242201\\
0.975571775429015	-0.16446177680834\\
0.976604880358111	-0.171346911507384\\
0.977288411053443	-0.178678859981665\\
0.977601326033031	-0.186458719640907\\
0.977519220635623	-0.194692243589618\\
0.977013149224172	-0.203388699602975\\
0.976048753650409	-0.212559844007609\\
0.974585542210078	-0.222219007989692\\
0.972576216648902	-0.232380260761747\\
0.969965981353522	-0.243057603125468\\
0.966691791725106	-0.254264142637042\\
0.962681512081745	-0.266011201169231\\
0.957852960646106	-0.278307304121314\\
0.952112822558047	-0.291156996428389\\
0.945355412866458	-0.304559423015273\\
0.937461270978256	-0.318506599605753\\
0.928295566587166	-0.332981282718821\\
0.917706294955367	-0.34795432357036\\
0.905522236725693	-0.36338135685302\\
0.89155065426707	-0.379198627970812\\
0.875574692916936	-0.395317694968563\\
0.857350451364069	-0.411618644100139\\
0.836603680776768	-0.42794131436639\\
0.813026067070305	-0.444073808942711\\
0.786271044859199	-0.459737232449454\\
0.755949085078022	-0.474565046231198\\
0.72162239089248	-0.488074516269009\\
0.682798928264884	-0.499626115944091\\
0.638925708276658	-0.508363753321968\\
0.589381227926346	-0.513122761203826\\
0.533466964480966	-0.512279840710781\\
0.470397805411754	-0.503488790673311\\
0.399291281323482	-0.483163037020644\\
0.319155452897268	-0.445292406129998\\
0.274036282569746	-0.415536998471938\\
0.228875284490919	-0.377961352325681\\
0.178056293588025	-0.322985149579327\\
0.127197316427285	-0.245934224926801\\
0.096727541204901	-0.178047907094137\\
0.0814882707916639	-0.128833386681334\\
0.0662462255770164	-0.0319480732142694\\
0.0652633297445542	0\\
0.0323015487303117	0\\
-0.226922089930545	0\\
-0.587248761060017	0\\
-0.929928974412894	0\\
-1.29124656113544	0\\
-1.68403554999352	0\\
-2.11689881687775	0\\
-2.59738504272603	0\\
-3.13299480180354	0\\
%-61.121263717454	0\\
inf	0\\
};
\addlegendentry{Locus $3$}

\end{axis}
\end{tikzpicture}%